\newif\ifDRAFT 
\DRAFTfalse
\newif\ifPODCSUB
\PODCSUBfalse
\ifPODCSUB
    \documentclass[acmsmall,nonacm,anonymous]{acmart}

    \AtBeginDocument{%
        }

    \setcopyright{acmcopyright}
    \copyrightyear{2024}
    \acmYear{2024}
    \acmDOI{XXXXXXX.XXXXXXX}

    \acmConference[PODC 2024]{ACM Symposium on Principles of Distributed Computing}{June 17--21,
      2024}{Nantes, France}
    \acmPrice{15.00}
    \acmISBN{978-1-4503-XXXX-X/18/06}

\else
    \documentclass[11pt]{article}
    \usepackage{amsfonts, amsmath, amssymb, amsthm}
    \usepackage[margin=1in]{geometry}
\fi

\usepackage{graphicx} % Required for inserting images
\usepackage{mdframed}
\usepackage{mathrsfs}  
\usepackage{algorithm}
\usepackage[noend]{algpseudocode}
\algrenewcommand\algorithmicrequire{\textbf{Input:}}
\algrenewcommand\algorithmicensure{\textbf{Output:}}
\usepackage{hyperref}
\usepackage{color}
\usepackage{array}
\usepackage{cleveref}
\usepackage[inkscapeformat=pdf]{svg}
\usepackage{multirow}
\usepackage{mathtools}

\usepackage{tabulary} %for tables
\theoremstyle{plain}
\newtheorem{theorem}{Theorem}[section]
\newtheorem{lemma}[theorem]{Lemma}
\newtheorem{corollary}[theorem]{Corollary}
\newtheorem{observation}[theorem]{Observation}

\theoremstyle{definition}
\newtheorem{definition}[theorem]{Definition}
\newtheorem*{remark}{Remark}%[section]
\crefname{equation}{Eqn.}{Eqns.}

\ifPODCSUB
\else
    \ifDRAFT
        \usepackage{lineno}
        \linenumbers
    \fi
\fi

\newcommand{\highlight}[1]{\textit{\textbf{#1}}}

\newcommand{\paren}[1]{\left( #1 \right)}
\newcommand{\ang}[1]{\left< #1 \right>}

\newcommand{\ceil}[1]{\lceil #1 \rceil}
\newcommand{\floor}[1]{\lfloor #1 \rfloor}
\newcommand{\sparen}[1]{\left[ #1 \right]}

\newcommand{\bydef}{\stackrel{\operatorname{def}}{=}}

\newcommand{\poly}{\operatorname{poly}}
\newcommand{\polylog}{\operatorname{polylog}}
\newcommand{\dist}{\operatorname{dist}}
\newcommand{\diam}{\operatorname{diam}}

\newcommand{\INDEX}{\textsc{Index}}
\newcommand{\bp}{\mathsf{bp}}
\newcommand{\id}{\mathsf{id}}
\newcommand{\cid}{\mathsf{cid}}

\newcommand{\pre}{\mathsf{pre}}
\newcommand{\post}{\mathsf{post}}
\newcommand{\eft}{\mathsf{E\text{-}faults}}

\newcommand{\UP}{\mathsf{UP}}
\newcommand{\NEXT}{\mathsf{NEXT}}

\newcommand{\FirstRecEdge}{\mathsf{FirstRecEdge}}

\usepackage[normalem]{ulem} % for crossing out sentences

\title{Connectivity Labeling in Faulty Colored Graphs}
\ifDRAFT
    \author{Anonymous Authors}
\else
    \author{
     Asaf Petruschka\thanks{Supported by the European Research Council (ERC) under the European Union’s Horizon 2020 research and innovation programme, grant agreement No. 949083, and by the Israeli Science Foundation (ISF), grant 2084/18.}
     \qquad
     Shay Sapir\thanks{This research was partially supported by the Israeli Council for Higher Education (CHE) via the Weizmann Data Science Research Center.}
     \qquad
     Elad Tzalik\\
     Weizmann Institute of Science
     \\ \texttt{\{asaf.petruschka,shay.sapir,elad.tzalik\}@weizmann.ac.il}}
\fi
\date{}

\begin{document}

\ifPODCSUB

\begin{abstract}

Fault-tolerant connectivity labelings are schemes that, given an $n$-vertex graph $G=(V,E)$ and a parameter $f$, produce succinct yet informative labels for the elements of the graph. Given \emph{only the labels} of two vertices $u,v$ and of the elements in a faulty-set $F$ with $|F|\leq f$, one can determine if $u,v$ are connected in $G-F$, the surviving graph after removing $F$. For the edge or vertex faults models, i.e., $F\subseteq E$ or $F \subseteq V$, a sequence of recent work established schemes with $\poly(f,\log n)$-bit labels for general graphs. This paper considers the \emph{color faults} model, recently introduced in the context of spanners [Petruschka, Sapir and Tzalik,~ITCS~'24], which accounts for known correlations between failures. Here, the edges (or vertices) of the input $G$ are arbitrarily colored, and the faulty elements in $F$ are colors; a failing color causes all edges (vertices) of that color to crash. While treating color faults by na\"{i}vly applying solutions for many failing edges or vertices is inefficient, the known correlations could potentially be exploited to provide better solutions.

Our main contribution is settling the label length complexity for connectivity under one color fault ($f=1$). The existing implicit solution, by black-box application of the state-of-the-art scheme for edge faults of [Dory and Parter, PODC '21], might yield labels of $\Omega(n)$ bits. We provide a deterministic scheme with labels of $\tilde{O}(\sqrt{n})$ bits in the worst case, and a matching lower bound. Moreover, our scheme is \emph{universally optimal}: even schemes tailored to handle only colorings of one specific graph topology (i.e., may store the topology ``for free'') cannot produce asymptotically smaller labels. We characterize the optimal length by a new graph parameter $\mathsf{bp}(G)$ called the \emph{ball packing number}. We further extend our labeling approach to yield a routing scheme avoiding a single forbidden color, with routing tables of size $\tilde{O}(\mathsf{bp}(G))$ bits. We also consider the \emph{centralized} setting, and show an $\tilde{O}(n)$-space oracle, answering connectivity queries under one color fault in $\tilde{O}(1)$ time. Curiously, by our results, no oracle with such space can be \emph{evenly} distributed as labels.

Turning to $f\geq 2$ color faults, we give a randomized labeling scheme with $\tilde{O}(n^{1-1/2^f})$-bit labels, along with a lower bound of $\Omega(n^{1-1/(f+1)})$ bits. For $f=2$, we make partial improvement by providing labels of $\tilde{O}(\mathrm{diam}(G)\sqrt{n})$ bits, and show that this scheme is (nearly) optimal when $\mathrm{diam}(G)=\tilde{O}(1)$.

Additionally, we present a general reduction from the above \emph{all-pairs} formulation of fault-tolerant connectivity labeling (in any fault model) to the \emph{single-source} variant, which could also be applicable for centralized oracles, streaming, or dynamic algorithms.

\end{abstract}

\begin{CCSXML}
<ccs2012>
   <concept>
       <concept_id>10003752.10003809.10003635</concept_id>
       <concept_desc>Theory of computation~Graph algorithms analysis</concept_desc>
       <concept_significance>500</concept_significance>
       </concept>
   <concept>
       <concept_id>10003752.10003809.10010031</concept_id>
       <concept_desc>Theory of computation~Data structures design and analysis</concept_desc>
       <concept_significance>500</concept_significance>
       </concept>
 </ccs2012>
\end{CCSXML}

\ccsdesc[500]{Theory of computation~Graph algorithms analysis}
\ccsdesc[500]{Theory of computation~Data structures design and analysis}

\keywords{labeling schemes, fault-tolerance}
    
    \maketitle
\else
    \maketitle
    \pagenumbering{gobble}
    
    \newpage
    \tableofcontents
    \newpage
    \pagenumbering{arabic}
\fi

\section{Introduction}

Labeling schemes are important distributed graph data structures with diverse applications in graph algorithms and distributed computing, concerned with assigning the vertices of a graph (and possibly also other elements, such as edges) with succinct yet informative \emph{labels}. 
Many real-life networks are often error-prone by nature, which motivates the study of fault-tolerant graph structures and services.
In a fault-tolerant connectivity labeling scheme, we are 
given an $n$-vertex graph $G=(V,E)$ and an integer $f$,
and should assign short labels to the elements of $G$, such that the following holds:
For every pair of vertices $u,v \in V$ and faulty-set $F$ 
with $|F|\leq f$,
% of at most $f$ failing elements,
one can determine if $u$ and $v$ are connected in $G-F$ by merely inspecting the labels of the elements in $\{u,v\}\cup F$.
The main complexity measure is the maximal \emph{label length} (in bits), while construction and query time are secondary measures.

The concept of edge/vertex-fault-tolerant labeling, aka \emph{forbidden set} labeling, was explicitly introduced by Courcelle and Twigg~\cite{CourcelleT07}.
Earlier work on fault-tolerant connectivity and distance labeling focused on graph families such as planar graphs and graphs with bounded treewidth or doubling dimension~\cite{CourcelleT07,CourcelleGKT08,AbrahamCG12,AbrahamCGP16}.
Up until recently, 
designing edge- or vertex-fault-tolerant connectivity labels for general graphs remained fairly open.
% For general graphs, 
% a sequence of recent works established schemes with $\poly(f,\log n)$-bit labels.
Dory and Parter~\cite{DoryP21} were the first to construct \emph{randomized} labeling schemes for connectivity under $f$ \emph{edge} faults,
where a query is answered correctly with high probability,%
% i.e., the label construction is randomized, and queries are answered correctly with high probability.%
\footnote{
    Throughout, the term with high probability~(w.h.p.)~stands for probability at least $1 - 1/n^{\alpha}$, where $\alpha>0$ is a constant that can be made arbitrarily large through increasing the relevant complexity measure by a constant factor.
}
with length of $O(\min \{f+\log n,\log^3 n\})$ bits.
% Their labels have length of $O(\min \{f+\log n,\log^3 n\})$ bits.
Izumi, Emek, Wadayama and Masuzawa~\cite{IzumiEWM23} 
derandomized this construction, showing \emph{deterministic} labels of $\tilde{O}(f^2)$ bits.%
\footnote{
    Throughout, the $\tilde{O}(\cdot)$ notations hides $\polylog(fn)$ factors.
}
Turning to $f$ \emph{vertex} faults,
Parter and Petruschka~\cite{ParterP22a} designed connectivity labels for $f\leq 2$ with $\tilde{O}(1)$ bits.
Very recently, Parter, Petruschka and Pettie~\cite{ParterPP23} provided 
a randomized scheme for $f$ vertex faults with $\tilde{O}(f^3)$ bits and a derandomized version with $\tilde{O}(f^7)$ bits, along with a lower bound of $\Omega(f)$ bits.

In this work, we consider labeling schemes for connectivity under \emph{color faults}, 
a model that was very recently introduced in the context of graph spanners \cite{PST24}, which intuitively accounts for known correlations between failures.
In this model, the edges or vertices of the input graph $G$ are arbitrarily partitioned into classes, or equivalently, associated with \emph{colors}, and a set $F$ of $f$ such color classes might fail.
A failing color causes all edges (vertices) of that
color to crash.
The survivable subgraph $G-F$ is formed by deleting every edge%
\footnote{In the edge-colored case, it is natural to consider multi-graphs, where parallel edges may have different colors.}
or vertex with color from $F$.
The scheme must assign labels to the vertices \emph{and to the colors} of $G$, so that a connectivity query $\ang{u,v,F}$ can be answered by inspecting only the labels of the vertices $u,v$ and of the colors in $F$.

This new notion generalizes edge/vertex fault-tolerant schemes, that are obtained in the special case when each edge or vertex has a unique color.
However, in the general case, even a single color fault may correspond to many and arbitrarily spread edge/vertex faults, which poses a major challenge.
Tackling this issue by naively applying the existing solutions for many individual edge/vertex faults (i.e., by letting the label of a color store all labels given to elements in its class) may result in very large labels of $\Omega(n)$ bits or more, even when $f=1$.
On a high level, this work shows that the correlation between the faulty edges/vertices, predetermined by the colors, can be used to construct much better solutions.

\paragraph{Related Work on Colored Graphs.}
Faulty colored classes have been used to model
Shared Risk Resource Groups (SRRG) in optical telecommunication networks, multi-layered networks, and various other practical contexts; see~\cite{CoudertDPRV07,Kuipers12,ZPT11} and the references therein.
Previous work mainly focused on centralized algorithms for colored variants of classical graph problems (and their hardness).
A notable such problem is diverse routing, where the goal is to find two (or more) color disjoint paths between two vertices \cite{Hu03,EGBERRL03,MENTA02}.
Another is the colored variant of minimum cut, known also as the \emph{hedge connectivity}, where the objective is to determine the minimum number of colors (aka hedges) whose removal disconnects the graph; see e.g.\ ~\cite{GhaffariKP17,XuS20,FoxPZ23}.

A different line of work focuses on distances to or between color classes, and specifically on (centralized) data structures that, given a query $\ang{v, c}$, report the closest $c$-colored vertex to $v$ in the graph, or the (approximate) distance from it~\cite{HermelinLWY11,Chechik12,LaishM17,GawrychowskiLMW18,Tsur18,EvaldFW21}.

\subsection{Our Results}

We initiate the study of fault-tolerant labeling schemes in colored graphs.
All of our results apply both to edge-colored and to vertex-colored (multi)-graphs.

\subsubsection{Single Color Fault ($f=1$)}

For $f = 1$, i.e., a single faulty color, we (nearly) settle the complexity of the problem, by showing a simple construction of labels with length $O(\sqrt{n} \log n)$ bits, along with a matching lower bound of $\Omega(\sqrt{n})$ bits.
% In fact, the guarantees of our scheme are stronger:
In fact, our scheme provides a strong beyond worst-case guarantee:
for every given graph $G$, the length of the assigned labels 
is (nearly) the best possible, even compared to schemes that are tailor-made to only handle colorings of the topology in $G$,%
\footnote{
    By the \emph{topology} of $G$, we mean the uncolored graph obtained from $G$ by ignoring the colors.
    Slightly abusing notation, we refer to this object as \emph{the graph topology} $G$, rather than the (colored) graph $G$.
}
or equivalently, are allowed to store the uncolored topology ``for free'' in all the labels.
Guarantees of this form, known as \emph{universal optimality}, have sparked major interest in the graph algorithms community, and particularly in recent years, following the influential work of Haeupler, Wajc and Zuzic~\cite{HaeuplerWZ21} in the distributed setting.%
\footnote{One cannot compete with a scheme that is optimal for the given graph \emph{and} its coloring (aka ``instance optimal''), as such a tailor-made scheme may store the entire colored graph ``for free'', and the labels merely need to specify the query.}
On an intuitive level, the universal optimality implies that even when restricting attention to any class of graphs, e.g.\ planar graphs, our scheme performs asymptotically as well as the optimal scheme for this specific class.

Our universally optimal labels are based on a new graph parameter called the \emph{ball packing number}, denoted by $\bp(G)$.
Disregarding minor nuances, $\bp(G)$ is the maximum integer $r$ such that one can fit $r$ disjoint balls of radius $r$ in the topology of $G$ (see formal definition in \Cref{par:ball_packing}). 
The ball packing number of an $n$-vertex graph is always at most $\sqrt{n}$, but often much smaller. For example, 
$\bp(G)$ is smaller than the \emph{diameter} of $G$.
In \Cref{sect:single-fault}, we show the following:

\begin{theorem}[$f=1$, informal]\label{thm:informal_one_fault}
There is a connectivity labeling scheme for one color fault, that for every $n$-vertex graph $G$,
assigns $O(\bp(G)\log n)$-bit labels. 
Moreover, $\Omega(\bp(G))$-bit labels are necessary, even for labeling schemes tailor-made for the topology of $G$,
i.e., where the uncolored topology is given in addition to the query labels.
\end{theorem}

The lower bound in \Cref{thm:informal_one_fault} is 
information-theoretic, obtained via communication complexity.
On a high level, we argue that one can encode $\bp(G)^2$ arbitrary bits by taking the following two steps:
(1) coloring the balls that certify the ball packing number  using $\bp(G)$ many colors, 
and (2) storing the labels of those colors and of additional $O(\bp(G))$ vertices.
The upper bound is based on observing that (when $G$ is connected) there is a subset $A$ of $O(\bp(G))$ vertices which is \emph{$O(\bp(G))$-ruling}: every vertex in $G$ has a path to $A$ of length $O(\bp(G))$.
Intuitively, this enables the label of a color $c$ to 
store the identifiers of the connected components of each $a\in A$ in $G-c$,
while the label of a vertex $v$ only stores 
the identifiers of
the connected components of $v$ 
with respect to faults of colors on its path to $A$.

Additionally, we consider the closely related problem of \emph{routing} messages over $G$ while avoiding any one forbidden color.
We extend our labeling approach, and show:

\begin{theorem}[Forbidden color routing, informal]
    There is a routing scheme for avoiding one forbidden color in a colored $n$-vertex graph $G$,
    with routing tables and labels of size $O(\bp(G)\log n)$ bits, and message header of size $O(\log n)$ bits.
\end{theorem}
See \Cref{sec:routing} for formal definitions, and a detailed discussion of our routing scheme.
We note that our routing scheme does not provide good distance (aka stretch) guarantees for the routing path, and optimizing it is an interesting direction for future work.

\medskip

We end our discussion for $f=1$ by considering \emph{centralized oracles} (data structures) for connectivity under a single color fault.
In this setting, one can utilize centralization to 
improve on the naive approach of
storing all labels.
We note that this problem can be solved using existing $O(n)$-space and $O(\log \log n)$-query time oracles for \emph{nearest colored ancestor} on trees \cite{MuthukrishnanM96,GawrychowskiLMW18}, yielding the same bounds for single color fault connectivity oracles.
Interestingly, our lower bound shows that oracles with such space cannot be evenly distributed into labels.
We are unaware of similar `gap' phenomena between the distributed and centralized variants of graph data structures.

\subsubsection{$f$ Color Faults}\label{subsec:intro_f_colors}

% We now discuss our results on labeling schemes for connectivity under $f \geq 2$ color faults.
It has been widely noted that in fault-tolerant settings, handling even two faults may be significantly more challenging than handling a single fault.
Such phenomena appeared, e.g., in distance oracles \cite{DuanP09a}, min-cut oracles \cite{BaswanaBP22}, reachability oracles \cite{Choudhary16} and distance preservers \cite{Parter15,GuptaK17,Parter20}.
In our case, this is manifested in generalized upper and lower bounds on the label length required to support $f$ color faults, exhibiting a gap when $f\geq 2$; our upper bound is roughly $\tilde{O}(n^{1-1/2^f})$ bits, while the lower bound is $\Omega(n^{1-1/(f+1)})$ bits (both equal $\tilde{\Theta}(\sqrt{n})$ when $f=1$).

\begin{theorem}[$f \geq 2$ upper bound, informal]\label{thm:f-upper-bound-informal}
    There is a randomized labeling scheme for connectivity under $f$ color faults with label length of $\min\{fn^{1-1/2^f}, n\} \cdot \polylog(fn)$ bits.
\end{theorem}

\begin{theorem}[$f \geq 2$ lower bound, informal]\label{thm:f-lower-bound-informal}
    A labeling scheme for connectivity under $f$ color faults must have label length of $\Omega(n^{1-1/(f+1)})$ bits for constant $f$, hence $\Omega(n^{1-o(1)})$ bits for $f = \omega(1)$.
\end{theorem}

The full discussion appears in \Cref{sect:f-faults}.
Apart from the gap between the bounds, 
there are a few more noteworthy differences from the case of a single color fault.

First, the scheme of \Cref{thm:f-upper-bound-informal} is \emph{randomized}, as opposed to the deterministic scheme for $f=1$ (\Cref{thm:informal_one_fault}).
Moreover, the construction is based on different techniques, combining three main ingredients:
(1) sparsification tools for colored graphs \cite{PST24},
(2) the (randomized) edge fault-tolerant labeling scheme of~\cite{DoryP21}, and (3) a recursive approach of~\cite{ParterP22a}.

Second, the lower bound of \Cref{thm:f-lower-bound-informal} is \emph{existential} (but still information-theoretic): it relies on choosing a fixed `worst-case' graph topology, and encoding information by coloring it and storing some of the resulting labels.
We further argue that this technique cannot yield a lower bound stronger than $\tilde{\Omega}(n^{1-1/(f+1)})$ bits.
This is due to the observation that a color whose label is not stored can be considered \emph{never faulty}, combined with the existence of efficient labeling schemes when the number of colors is small.
A detailed discussion of this barrier appears in \Cref{sect:f-faults-lower-bound}.

\medskip
Curiously, in the seemingly unrelated problem of small-size fault-tolerant distance preserves (FT-BFS) introduced by Parter and Peleg~\cite{ParterP13}, there is a similar gap in the known bounds for $f\geq 3$, of $O(n^{2-1/2^f})$ and $\Omega(n^{2-1/(f+1)})$ edges~\cite{Parter15,BodwinGPW17}.
Notably, for the case of $f=2$, Parter~\cite{Parter15} provided a tight upper bound of $O(n^{5/3})$, later simplified by Gupta and Khan~\cite{GuptaK17}.
% The case of $f\geq 3$ is still open.
Revealing connections between FT-BFS structures and the labels problem of this paper is an intriguing direction for future work.

\subsubsection{Two Color Faults and Graphs of Bounded Diameter}
For the special case of two color faults, we provide another scheme, with label length of $\tilde{O}(D \sqrt{n})$ bits for graphs of diameter at most $D$.

\begin{theorem}[$f=2$ upper bound, informal]
    There is a labeling scheme for connectivity under two color faults with label length of $\tilde{O}(D \sqrt{n})$ bits.
\end{theorem}

This beats the general scheme %e.g.\ 
when $D = O(n^{1/4 - \epsilon})$, and demonstrates that the existential $\Omega(n^{2/3})$ lower bound does not apply to graphs with diameter $D = O(n^{1/6 - \epsilon})$.
Further, this scheme is existentially optimal (up to logarithmic factors) for graphs with $D= \tilde{O}(1)$. 
We hope this construction, found in \Cref{sect:two-fault-diameter}, could serve as a stepping stone towards closing the current gap between our bounds, and towards generalizing $\bp(G)$ 
for the case of $f=2$.

\Cref{tbl:results_table} summarizes our main results on connectivity labeling under color faults.

\renewcommand{\arraystretch}{1.25} % Default value: 1
\begin{table}[t]
\caption{A summary of our results on $f$ color fault-tolerant connectivity labeling schemes. The table shows the provided length bounds (in bits) for such schemes.} \label{tbl:results_table}
\centering
\medskip
\begin{tabular}{|c|cc|cc|}
\hline
\textbf{No. faults}    & \multicolumn{2}{c|}{\textbf{Upper bound}}                               & \multicolumn{2}{c|}{\textbf{Lower bound}}                                      \\ \hline\hline
$f=1$                  & \multicolumn{1}{c|}{$\tilde{O}(\sqrt{n})$}       & Thm \ref{thm:single-fault-upper-bound}                 & \multicolumn{1}{c|}{$\Omega(\sqrt{n})$}                 & Thm  \ref{thm:single-fault-lower-bound}                \\ \hline
\multirow{2}{*}{$f=2$} & \multicolumn{1}{c|}{$\tilde{O}( \diam(G)  \sqrt{n}  )$} & Thm \ref{thm:two-fault-upper-bound}                 & \multicolumn{1}{c|}{\multirow{2}{*}{$\Omega(n^{2/3})$}} & \multirow{4}{*}{Thm \ref{thm:f-faults-lower-bound}} \\ \cline{2-3}
                       & \multicolumn{1}{c|}{$\tilde{O}( n^{3/4} )$}      & \multirow{3}{*}{Thm \ref{thm:f-faults-upper-bound}} & \multicolumn{1}{c|}{}                                   &                      \\ \cline{1-2} \cline{4-4}
$f=O(1)$               & \multicolumn{1}{c|}{$\tilde{O}(n^{1-1/2^f})$}    &                      & \multicolumn{1}{c|}{$\Omega(n^{1-1/(f+1)})$}            &                      \\ \cline{1-2} \cline{4-4}
$f=\omega(1)$          & \multicolumn{1}{c|}{$\tilde{O}(n)$}              &                      & \multicolumn{1}{c|}{$\Omega(n^{1-o(1)})$}               &                     \\
\hline
\end{tabular}
\end{table}
\renewcommand{\arraystretch}{1} % Default value: 1

\subsubsection{Equivalence Between All-Pairs and Single-Source Connectivity}

In the \emph{single-source} variant of fault-tolerant connectivity, 
given are an $n$-vertex graph $G$ with a designated \emph{source vertex} $s$, and an integer $f$.
It is then required to
support queries of the form $\ang{u,F}$, where $u\in V$ and $F$ is a faulty-set of size at most $f$, by reporting whether $u$ is connected to $s$ in $G-F$.
Here, and throughout this discussion, we do not care about the type of faulty elements; these could be edges, vertices or colors.
For concreteness, we focus our discussion on labeling schemes, although it applies more generally to other models, e.g., centralized oracles, streaming, and dynamic algorithms.
Clearly, every labeling scheme for \emph{all-pairs} fault-tolerant connectivity can be transformed into a single-source variant by including $s$'s label in all other labels, which at most doubles the label length.
We consider the converse direction, and show that a single-source scheme can be used as a black-box to obtain an all-pairs scheme with only a small overhead in length.

\begin{theorem}[Single-source reduction, informal]
    % Let $f \geq 1$.
    Suppose there is a single-source $f$ fault-tolerant connectivity labeling scheme using labels of at most $b(n, f)$ bits.
    Then, there is an all-pairs $f$ fault-tolerant connectivity labeling scheme with $\tilde{O}(b(n+1, f))$-bit labels. 
\end{theorem}

The reduction is based on the following idea.
Suppose we add a new source vertex $s$ to $G$, and include each edge from $s$ to the other vertices independently with  probability $p$.
Given a query $\ang{u,v,F}$,
if $u,v$ are originally connected in $G - F$, they must agree on connectivity to the new source $s$, regardless of $p$.
However, if $u,v$ are disconnected in $G-F$, and $p$ is such that  $1/p$ is roughly the size of $u$'s connected component in $G-F$, then with constant probability, $u$ and $v$ will disagree on connectivity to $s$.
The full proof appears in \Cref{sect:single-source-reduction}.

\section{Preliminaries}\label{sect:prelim}

\paragraph{Colored Graphs.}
Throughout, we denote the given input graph by $G$, which is an undirected 
graph with $n$ vertices $V = V(G)$, and $m$ edges $E = E(G)$.
The graph $G$ may be a multi-graph, i.e., there may be several different edges with the same endpoints (parallel edges).
The edges or the vertices of $G$ are each given a color
from a set of $C$ possible colors.
% from $[C] = \{1, \dots, C\}$.
The coloring is \emph{arbitrary}; there are no `legality' restrictions (e.g., edges sharing an endpoint may have the same color).
Without loss of generality, we sometimes assume that $C \leq \max\{m,n\}$, and that the set of colors is $[C]$.
For a (faulty) subset of colors $F$, we denote by $G-F$ the subgraph of $G$ where all edges (or vertices) with color from $F$ are deleted.
When $F$ is a singleton $F = \{c\}$, we use the shorthand $G - c$.

In some cases, we refer only to the \emph{topology} of the graph, and ignore the coloring.
Put differently, we sometimes consider the family of inputs given by all different colorings of a fixed graph.
This object is referred to as the graph \emph{topology} $G$, rather than the graph $G$. We denote by $\dist_G (u,v)$ the number of edges in a $u$-$v$ shortest path (and $\infty$ if no such path exist).
For a non-empty $A \subseteq V$, the distance from $u \in V$ to $A$ is defined as $\dist_G(u,A) = \min\{\dist_G(u,a) \mid a \in A\}$.

\medskip
Our presentation focuses, somewhat arbitrarily, on the \emph{edge-colored} case; throughout, this case is assumed to hold unless we explicitly state otherwise.
This is justified by the following discussion.

\paragraph{Vertex vs.\ Edge Colorings.}
An edge-colored instance can be reduced to a vertex-colored one, and vice versa, by subdividing each edge%
\footnote{Throughout, we slightly abuse notation and write $e = \{u,v\}$ to say that $e$ has endpoints $u,v$, even though there might be several different edges with these endpoints.} 
$e = \{u,v\}$ into two edges $\{u,x_e\}$ and $\{x_e, v\}$, where $x_e$ is a new vertex.
If the original instance has edge colors, we give the new instance vertex colors, by coloring each new vertex $x_e$ with the original color of the edge $e$.
(The original vertices get a new `never-failing' color.)
For the other direction, we color each of $\{u,x_e\}$ and $\{x_e, v\}$ by the color of the original vertex incident to it, i.e., $\{u,x_e\}$ gets $u$'s color, and $\{x_e, v\}$ gets $v$'s color.

These easy reductions increase the number of vertices to $n + m$, which a prioi might seem problematic.
However, as shown by~\cite{PST24}, 
given any fixed (constant) bound $f$ on the number of faulty colors, one can replace a given input instance (either vertex- or edge-colored) by an equivalent sparse subgraph with only $\tilde{O}(n)$ edges, that has the same connectivity as the original graph under any set of at most $f$ color faults.
So, by sparsifying before applying the reduction, the number of vertices increases only to $\tilde{O}(n)$.
Moreover, all our results translate rather seamlessly between the edge-colored and the vertex-colored cases, even without the general reductions presented above.

\paragraph{Vertex and Component IDs.}
We assume w.l.o.g.\ that the vertices have unique $O(\log n)$-bit identifiers from $[n]$, where $\id(v)$ denotes the identifier of $v \in V$.
Using these, we define identifiers for connected components in subgraphs of $G$, as follows.
When $G'$ is a subgraph of $G$ and $v \in V(G')$, we define $\cid(v,G') = \min\{\id(u) \mid \text{$u,v$ connected in $G'$}\}$.
This ensures $\cid(u,G') = \cid(v,G')$ iff $u,v$ are in the same connected component in $G'$.
Therefore, if one can compute $\cid(v,G-F)$ from the labels of $v,F$, then, using the same labels, one can answer connectivity queries subject to faults.

\paragraph{Indexing Lower Bound.}
Our lower bounds rely on the classic \emph{indexing} lower bound from communication complexity.
In the one-way communication problem $\INDEX(N)$,
Alice holds a string $x \in \{0,1\}^N$, and Bob holds an index $i \in [0,N-1]$.
The goal is for Alice to send a message to Bob, such that Bob can recover $x_i$, the $i$-th bit of $x$.
Crucially, the communication is one-way; Bob cannot send any message to Alice.
The protocols are allowed to be randomized, in which case both Alice and Bob have access to a public random string.
The following lower bound on the number of bits Alice is required to send is well-known (see \cite{kushilevitz_nisan_1996, KremerNR99, JayramKS08}).

\begin{lemma}[Indexing Lower Bound~\cite{KremerNR99}]\label{lem:index_LB}
Every one-way communication protocol (even with shared randomness) for $\INDEX(N)$ must use $\Omega(N)$ bits of communication.
\end{lemma}

\section{Single Color Fault}\label{sect:single-fault}

In this section, we study the connectivity problem under one color fault. That is, given two vertices $u,v$ and a faulty color $c$, one should be able to determine if $u,v$ are connected in $G-c$. 
In \Cref{sect:1-fault-main-upper,sect:1-fault-lower-bound} we focus on labeling schemes, and provide universally optimal upper and lower bounds.
In \Cref{sect:1-fault-centralized} we change gears and provide centralized oracles for this problem.

\subsection{Our Scheme and the Ball Packing Number}\label{sect:1-fault-main-upper}

We first show a scheme that works when $G$ is connected.%
\footnote{
    Connectivity cannot be assumed without losing generality, because of the color labels.
    A color gets \emph{only one} label, which should support connectivity queries in \emph{every} connected component of the input graph.
}
Later, 
% in \Cref{sect:1-fault-upper-bound},
we show how to remove this assumption.
Consider the following procedure:
starting from an arbitrary vertex $a_0$,
iteratively choose a vertex $a_i$ which satisfies
\[
\dist_G ( a_i, \{a_0, \dots, a_{i-1}\} ) = i,
\]
until no such vertex exists.
Suppose the procedure halts at the $k$-th iteration,
with the set of chosen vertices $A = \{a_0, \dots, a_{k-1}\}$.
Then every vertex $v \in V$ has distance less than $k$ from $A$.
We use $A$ to construct $O(k \log n)$-bit labels, as follows.

\begin{itemize}
    \item \highlight{Label $L(c)$ of color $c \in [C]$:} \
    For every $a \in A$, store $\cid(a, G-c)$.
    
    \item \highlight{Label $L(v)$ of vertex $v\in V$:} \
    Let $P(v)$ be a shortest path connecting $v$ to $A$, and let $a(v)$ be its endpoint in $A$.
    For every color $c$ present in $P(v)$, store $\cid(v, G-c)$.
    Also, store $\id(a(v))$.
\end{itemize}
Answering queries is straightforward as given $L(v)$ and $L(c)$, one can readily compute $\cid(v, G-c)$:
If the color $c$ appears on the path $P(v)$, then $\cid(v, G-c)$ is found in $L(v)$.
Otherwise, $P(v)$ connects between $v$ and $a(v)$ in $G-c$, hence $\cid(v, G-c) = \cid(a(v), G-c)$, and the latter is stored in $L(c)$.

The labels have length of $O(\sqrt{n} \log n)$ bits, as follows.
Consider the $A$-vertices chosen at iteration $\ceil{k/2}$ or later.
By construction, each of these $\floor{k/2}$ vertices is at distance at least $\ceil{k/2}$ from all others.
Hence, the balls of radius $\floor{k/4}$ (in the metric induced by $G$) centered at these vertices are disjoint, and each such ball contains at least $\floor{k/4}$ vertices.
Thus, $\floor{k/2} \cdot \floor{k/4} \leq n$, so $k = O(\sqrt{n})$.

\medskip
The length of the labels assigned by this simple scheme turns out to be not only \emph{existentially optimal}, but also \emph{universally 
optimal} (both up to a factor of $\log n$).
By existential optimality, we mean that every labeling scheme for connectivity under one color fault must have $\Omega(\sqrt{n})$-bit labels on some \emph{worst-case colored graph $G$}.
The stronger universal optimality means that for \emph{every graph topology $G$}, every such labeling scheme, even tailor-made for $G$, must assign $\Omega(k)$-bit labels (for some coloring of $G$).

\paragraph{The Ball-Packing Number.}\label{par:ball_packing}
To prove the aforementioned universal optimality of our scheme, we introduce a graph parameter called the \emph{ball-packing number}.
As the name suggests, this parameter concerns packing disjoint balls in the metric induced by the graph topology $G$.
Its relation to faulty-color connectivity is hinted by the previous analysis using a ``ball packing argument'' to obtain the $\tilde{O}(\sqrt{n})$ bound.
We next give the formal definitions and some immediate observations.

\begin{definition}[Proper $r$-ball]\label{def:proper_ball}
    For every integer $r \geq 0$, the \emph{$r$-ball} in $G$ centered at $v \in V(G)$, denoted $B_G (v,r)$, consists of all vertices of distance at most $r$ from $v$.
    That is,
    \[
    B_G (v, r) \bydef \{u \in V(G) \mid \dist_G (v,u) \leq r \}.
    \]
    An $r$-ball is called \emph{proper} if 
    there exists $u\in B_G(v,r)$ that realizes the radius, i.e., $\dist_G(u,v)=r$.%
    \footnote{
        If the distance $r$ from $v$ is not realized, then 
        there exists $r'<r$ such that
        the distance $r'$ \emph{is} realized, and
        $B_G (v, r') = B_G (v, r)$ \emph{as sets of vertices}.
        So, whether $B(v,r)$ is proper depends not only on the set of vertices in this ball, but also on the specified parameter $r$.
    }
\end{definition}

\begin{observation}\label{obs:proper-balls}
    If $r \leq \dist_G (u,v) < \infty$, then $B_G (u,r)$ and $B_G (v,r)$ are proper $r$-balls.
\end{observation}

\begin{definition}[Ball-packing number]\label{def:ball_packing}
    The \emph{ball-packing number} of $G$, denoted $\bp(G)$, is the maximum integer $r$ such that there exist at least $r$ vertex-disjoint proper $r$-balls in $G$.
\end{definition}

\begin{observation}\label{obs:bp-at-most-sqrt-n}
(i) For every $n$-vertex graph $G$, $\bp(G) \leq \sqrt{n}$.
(ii) For some graphs $G$, we also have $\bp(G) = \Omega(\sqrt{n})$ (e.g., when $G$ is a path).
\end{observation}

\paragraph{A Ball-Packing Upper Bound.}
% Consider the simple scheme described above.
Our length analysis for the above scheme in fact showed the existence of at least $\floor{k/2}$ disjoint and proper $\floor{k/4}$-balls, implying that $k = O(\bp(G))$ by \Cref{def:ball_packing}.
Minor adaptations to this scheme to handle several connected components in $G$ yields the following theorem, whose proof is deferred to \Cref{sect:1-fault-upper-bound}.

\begin{theorem}\label{thm:single-fault-upper-bound}\label{THM:SFUB}
    There is a deterministic labeling scheme for connectivity under one color fault that, when given as input an $n$-vertex graph $G$, assigns labels of length $O(\bp(G) \log n)$ bits.
    The query time is $O(1)$ (in the RAM model).
\end{theorem}
\begin{remark}
    By \Cref{obs:bp-at-most-sqrt-n}(i), the label length is always bounded by $O(\sqrt{n} \log n)$ bits.
\end{remark}

\def\SINGLEFAULTUPPER{
% \begin{proof}[Proof of \Cref{thm:single-fault-upper-bound}]
\highlight{Labeling.}
The labeling procedure is presented as \Cref{alg:single-fault-labels}.

\begin{algorithm}
\caption{Labeling for one color fault}\label{alg:single-fault-labels}
% \textbf{Input:} Colored graph $G$. \\
% \textbf{Output:} Labels $L(v)$ for each vertex $v \in V$, and $L(c)$ for each color $c$.
\begin{algorithmic}[1]
\Require Colored graph $G$.
\Ensure Labels $L(v)$ for each vertex $v \in V$, and $L(c)$ for each color $c$.
\State $A_0 \gets \{a \in V \mid \id(a) = \cid(a,G)\}$ \Comment{vertices w/ min $\id$ in each connected component of $G$}
\State $A \gets \emptyset$
\State $i \gets 1$
\While{there exists vertex $a_i \in V$ with $\dist_G (a_i, A_0 \cup A) = i$}
    \State $A \gets A \cup \{a_i\}$
    \State $i \gets i+1$
\EndWhile
\For{each vertex $v \in V$}
    \Comment{create the label $L(v)$}
    \State $a(v) \gets $ a closest vertex to $v$ from $A_0 \cup A$ in $G$
    \State $P(v) \gets$ a shortest $v$-to-$a(v)$ path in $G$
    \State \textbf{store} in $L(v)$ the id of $a(v)$, $\id(a(v))$ 
    \State \textbf{store} in $L(v)$ a dictionary that maps key $d$ to value $\cid(v, G-d)$, for every color $d$ on $P(v)$
\EndFor
\For{each color $c$}
    \Comment{create the label $L(c)$}
    \State \textbf{store} in $L(c)$ the name of the color $c$
    \State \textbf{store} in $L(c)$ a dictionary that maps key $\id(a)$ to value $\cid(a,G-c)$, for each $a \in A$
\EndFor
\end{algorithmic}
\end{algorithm}

\medskip
\highlight{Answering queries.}
As before, it suffices to show that one can report $\cid(v, G-c)$ merely from the labels $L(v)$ and $L(c)$.
This is done as follows.
If the key $c$ appears in the dictionary of $L(v)$, then $\cid(v,G-c)$ is the corresponding value, and we are done.
Otherwise, $P(v)$ does not contain the color $c$, so it connects $v$ to $a(v)$ in $G-c$, and therefore $\cid(v,G-c)=\cid(a(v),G-c)$.
Recall that $a(v) \in A_0 \cup A$.
If the key $\id(a(v))$ appears in the dictionary of $L(c)$, then $\cid(a(v),G-c)$ is the corresponding value, and we are done.
Otherwise, it must be that $a(v) \in A_0$. Recall that $A_0$ contains vertices having minimum $\id$ in their connected component in $G$.
This implies that $\cid(a(v), G-c) = \id(a(v))$, and the latter is stored in $L(v)$.

\medskip
\highlight{Length analysis.}
Let $k$ be the halting iteration of the while loop in \Cref{alg:single-fault-labels} (line 4).
Then $|A| = k-1$, so the length of each color label is $O(k \log n)$ bits.
Also, by the while condition, each vertex $v \in V$ has distance less than $k$ from $A_0 \cup A$ in $G$, so $P(v)$ contains less than $k$ edges, hence also less than $k$ colors. 
Therefore, the length of $L(v)$ is also $O(k \log n)$ bits.

We now prove that $k = O(\bp(G))$.
Let $A' = \{a_i \mid i \geq \ceil{k/2}\}$ be the ``second half'' of chosen $A$-vertices.
Note that $|A'|\geq k/2$.
For each $i = 1, \dots k-1$, denote by $A_i$ the state of set $A$ right after the $i$-th iteration.
If $a_i, a_j \in A'$ with $i > j$, then $a_j \in A_{i-1}$, hence 
\[
\dist_G (a_i, a_j) \geq \dist_G (a_i, A_0 \cup A_{i-1}) = i > \ceil{k/2}.
\]
Therefore, $B(a_i, \floor{k/4})$ and $B(a_j, \floor{k/4})$ are disjoint.
Also, if $a_i \in A'$, then $\dist_G (a_i, A_0 \cup A) = i>\floor{k/4}$,  %with $\floor{k/4} < i < \infty$, 
so \Cref{obs:proper-balls} implies that $B(a_i, \floor{k/4})$ is proper.
Thus, the collection of $\floor{k/4}$-balls centered at the $A'$ vertices certifies that $\floor{k/4} \leq \bp(G)$, %hence $k = O(\bp(G))$.
which concludes the proof.
% \end{proof}

\begin{remark}
The proof extends seamlessly to vertex-colored graphs.
It is worth noting that if the vertex $a(v)$ has the color $c$, then $\cid(v,G-c)$ is stored in $L(v)$ (since $a(v)\in P(v)$).
\end{remark}
}%\SINGLEFAULTUPPER

\subsection{A Ball-Packing Lower Bound}\label{sect:1-fault-lower-bound}
We now show an $\Omega(\bp(G))$ bound on the maximal label length. %bits are necessary for the labels.

\begin{theorem}%[Ball-Packing Lower Bound]
\label{thm:single-fault-lower-bound}
    Let $G$ be a graph topology.
    Suppose there is a 
    (possibly randomized) labeling scheme for connectivity under one color fault, that assigns labels of length at most $b$ bits for every coloring of $G$.
    % must have
    Then 
    $b = \Omega(\bp(G))$.
\end{theorem}

\begin{remark}
    By the above theorem and \Cref{obs:bp-at-most-sqrt-n}(ii), every labeling scheme for all topologies must assign $\Omega(\sqrt{n})$-bit labels on some input, which proves 
    \Cref{thm:f-lower-bound-informal}
    for the special case $f=1$.
\end{remark}

\begin{proof}[Proof of \Cref{thm:single-fault-lower-bound}]
    Denote $r = \bp(G)$.
    The proof uses the labeling scheme and the graph topology $G$ to construct a communication protocol for $\INDEX(r^2)$.
    Let $x = x_0 x_1 \cdots x_{r^2-1}$ be the input string given to Alice,
    where each $x_i \in \{0,1\}$.
    Let $i^*$ be the index given to Bob, where $0 \leq i^* \leq r^2-1$.
    On a high level, the communication protocol works as follows.
    Both Alice and Bob know the (uncolored) graph topology $G$ in advance, as part of the protocol.
    Alice colors the edges of her copy of $G$ according to her input $x$,
    and applies the labeling scheme to compute labels for the vertices and colors.
    She then sends $O(r)$ such labels to Bob, and he recovers $x_{i^*}$ by using the labels to answer a connectivity query in the colored graph.
    As the total number of sent bits is $O(b \cdot r)$, it follows by \Cref{lem:index_LB}  that $b\cdot r = \Omega(r^2)$, and hence $b = \Omega(r) = \Omega(\bp(G))$.
    The rest of this proof is devoted to the full description of the protocol.

    In order to color $G$, Alice does the following. She uses
    % We start with the coloring of $G$ by Alice.
    the color-palette %used by Alice is 
    $\{0,1,\dots, r-1\} \cup \{\perp\}$, where
    the symbol $\perp$ is used instead of $r$ to stress that $\perp$ is a special  \emph{never failing color} in the protocol.
    Let $v_0, v_1 \dots, v_{r-1}$ be centers of $r$ disjoint proper $r$-balls in $G$, which exist by \Cref{def:ball_packing} of Ball-Packing, and since $r = \bp(G)$.
    For every $k,l \in [0, r)$,
    define
    \[
    E_{k, l}
    \bydef
    \big\{
        \{u,w\} \in E
        \mid \text{$\dist_G(v_k, u) = l$ and $\dist_G (v_k, w) = l+1$}
    \big\}.
    \]
    In other words, $E_{k,l}$ is the set of edges connecting layers $l$ and $l+1$ of the $k$-th ball $B(v_k, r)$.
    As the layers in a ball are disjoint, and the balls themselves are disjoint, the sets $\{E_{k,l}\}_{k,l}$ are mutually disjoint.
    Alice colors these edge-sets by the following rule:
    For every $i\in [0,r^2-1]$, she decomposes it as $i = kr + l$ with $l,k\in [0, r)$.
    If 
    % If, for $i = kr + l$, it holds that 
    $x_i = 1$, the edges in $E_{k,l}$ get the color $l$.
    Otherwise, when $x_i = 0$, these edges get the null-color $\perp$.
    Every additional edge in $G$, outside of the sets $\{E_{k,l}\}_{k,l}$, is also colored by $\perp$.    
    The purpose of this coloring is to ensure the following property, for $k, l \in [0,r)$ and $i = kr+l$: 
    If $x_i = 0$, then (the induced graph on) $B_G(v_k, r)$ does not contain any $l$-colored edges and its vertices are connected in $G-l$.
    However, if $x_i = 1$, then $E_{k,l}$ is colored by $l$, hence in $G-l$, $v_k$ is disconnected from every $u$ for which $\dist_G (u,v_k) > l$.
    
    Next, we describe the message sent by Alice.
    For $0 \leq k \leq r-1$,
    let $u_k \in V$ with $\dist_G (u_k, v_k) = r$,
    which exists by \Cref{def:proper_ball}, as $B_G(v_k, r)$ is a \emph{proper} $r$-ball.
    Alice applies the labeling scheme on the colored $G$,
    and sends to Bob the labels of the vertices $v_0, \dots, v_{r-1}, u_0, \dots, u_{r-1}$, and of the colors $0, \dots, r-1$.
    This amounts to $3r$ labels.

    Finally, we describe Bob's strategy.
    He decomposes $i^*$ as $i^* = k^*  r + l^*$
    with $k^*, l^* \in [0,r)$,
    and uses the labels of $v_{k^*}, u_{k^*}, l^*$ to query the connectivity of $v_{k^*}$ and $u_{k^*}$ in $G - l^*$.
    If the answer is \emph{disconnected}, Bob determines that $x_{i^*} = 1$, and if it is \emph{connected}, he determines that $x_{i^*} = 0$.
    By the previously described property of the coloring, Bob indeed recovers $x_{i^*}$ correctly.
    Thus, this protocol solves $\INDEX(r^2)$, which concludes the proof.

    This proof extends quite easily to \emph{vertex-colored} graphs; Alice can color the vertices in the $l$-th layer of $B(v_k,r)$ instead of the edges $E_{k,l}$.
\end{proof}

\subsection{Centralized Oracles and Nearest Colored Ancestors}\label{sect:1-fault-centralized}
In the \emph{centralized} setting of oracles for connectivity under one color fault, the objective is to preprocess the colored graph $G$ into a low-space centralized data structure (oracle) that, when queried with (the names/$\id$s of) two vertices $u,v \in V$ and a color $c$, can quickly report if $u$ and $v$ are connected in $G-c$.
The labeling scheme of \Cref{thm:single-fault-upper-bound} implies such a data structure with $O(n^{1.5})$ space and $O(1)$ query time.%
\footnote{
    The data structure stores all vertex labels, and the labels of all colors that appear in
    some fixed maximal spanning forest $T$ of $G$.
    We can ignore all other colors, as their failure does not change the connectivity in $G$.
}
(The bounds for centralized data structures are in the standard RAM model with $\Theta(\log n)$-bit words.)
By the lower bound of \Cref{thm:single-fault-lower-bound}, such a data structure with space $o(n^{1.5})$ cannot be ``evenly distributed'' into labels.

However, utilizing centralization, we can achieve $O(n)$ space with only $O(\log \log n)$ query time.
This is obtained by a reduction to the \emph{nearest colored ancestor} problem, studied by Muthukrishnan and M{\"{u}}ller~\cite{MuthukrishnanM96}
and by Gawrychowski, Landau, Mozes and Weimann \cite{GawrychowskiLMW18}.
They showed that a rooted $n$-vertex forest with colored vertices can be processed into an $O(n)$-space data structure, that given a vertex $v$ and a color $c$, returns the nearest $c$-colored ancestor of $v$ (or reports that none exist) in $O(\log \log n)$ time.
The reduction is as follows.
Choose a maximal spanning forest $T$ for $G$, and root each tree of the forest in the vertex with minimum $\id$.
For each vertex $u \in V$, assign it with the color $d$ of the edge connecting $u$ to its parent in $T$. Additionally, store $\cid(u, G-d)$ in the vertex $u$.
(The roots get a null-color and store their $\id$s, which are also their $\cid$s in every subgraph of $G$.)
Now, construct a nearest color ancestor data structure for $T$ as in~\cite{MuthukrishnanM96,GawrychowskiLMW18}.
Given a query $v \in V$ and color $c$, we can find the nearest $c$-colored ancestor $w$ of $v$ in $O(\log\log n)$ time.
As $w$ is nearest, the $T$-path from $v$ to $w$ in $T$ does not contain $c$-colored edges, implying that $\cid(v,G-c) = \cid(w,G-c)$, and the latter is stored at $w$.
(If no such $c$-colored ancestor exists, take $w$ as the root, and proceed similarly.)
Given $u,v \in V$ and color $c$, apply the above procedure twice, and determine the connectivity of $u,v$ in $G-c$ by comparing their $\cid$s, within $O(\log \log n)$ time.
We therefore get:
\begin{theorem}
    Every colored $n$-vertex graph $G$ can be processed into 
    an $O(n)$-space
    centralized oracle that 
    given a query of $u,v \in V$ and color $c$, reports if $u,v$ are connected in $G-c$ in $O(\log \log n)$ time.
\end{theorem}

The reduction raises an alternative approach for constructing connectivity labels for one color fault,
via providing a labeling scheme for the nearest colored ancestor problem.
In \Cref{sect:nearest-colored-ancestor-labels} we show that indeed, such a scheme with $\tilde{O}(\sqrt{n})$-bit labels exists.

\section{$f$ Color Faults}\label{sect:f-faults}

In this section, we connectivity labeling under (at most) $f$ color faults, for arbitrary $f$.

\subsection{Upper Bound}\label{sect:f-fault-upper-bound}

We provide two labeling schemes for connectivity under $f$ color faults: the first is better for small $f = o(\log \log n)$, and the second is good for larger values of $f$.
The following theorem is obtained by combining the two:

\begin{theorem}\label{thm:f-faults-upper-bound}
    Let $f \geq 1$.
    There is a randomized labeling scheme for connectivity under $\leq f$ color faults, assigning labels of length $O(\min\{f n^{1-1/2^f}, n\} \cdot \polylog(fn))$ bits on colored $n$-vertex graphs.
\end{theorem}

We state two ingredients required by our scheme.
The first is \emph{color fault-tolerant connectivity certificates}, recently constructed by \cite{PST24}.
The second is connectivity labels for \emph{edge faults} by Dory and Parter \cite{DoryP21}:

\begin{theorem}[Color fault-tolerant connectivity certificates~\protect{\cite[Theorem 21]{PST24}}]\label{lem:conn-cert}
    Given a colored $n$-vertex graph $G$, one can compute in polynomial time a subgraph $H$ with $O(fn \log n)$ edges, 
    that is an $f$-color fault-tolerant connectivity certificate: for all $u,v\in V$ and sets $F$ of at most $f$ colors, $u,v$ are connected in $G-F$ iff they are connected in $H-F$.
\end{theorem}

\begin{theorem}[Connectivity labels for edge faults \protect{\cite[Theorem 3.7]{DoryP21}}]\label{thm:edge-fault-lables}
    There exists a randomized labeling scheme that, when given a multi-graph $G=(V,E)$ with $n$ vertices and $m$ edges, assigns labels of $O(\log^3 n + \log m)$ bits to $V \cup E$,
    such that given the labels of any two vertices $u,v \in V$ and of the edges in $E' \subseteq E$, one can correctly determine, with high probability, if $u$ and $v$ are connected in $G - E'$. 
    Note that the label length is independent of $|E'|$, the number of faulty edges.
\end{theorem}

\subsubsection{Labeling Scheme for $f = o(\log \log n)$}

Parter and Petruschka~\cite{ParterP22a} gave a recursive construction of labels for $f$ \emph{vertex} faults, by combining the sparse vertex-connectivity certificates of \cite{NagamochiI92} with the labels for edge faults of \cite{DoryP21}.
The sparsification of \Cref{lem:conn-cert} allows us to extend this technique to handle color faults.

\begin{lemma}\label{lem:small-f-upper-bound}
    There is a randomized labeling scheme for connectivity under $\leq f$ color faults, assigning labels of length ${O}(f n^{1-1/2^f} \polylog(fn))$ bits on $n$-vertex graphs.
\end{lemma}

    The idea is to construct labels for $f$ color faults by combining the labels for edge faults of Dory-Parter \cite{DoryP21} (\Cref{thm:edge-fault-lables}) with recursively defined labels for $f-1$ faults.
    To this end, we classify the colors according to their \emph{prevalence} in the given input graph $G$.
    Let $\mathcal{H}$ be the set of \emph{high prevalence} colors consisting of every color $c$ that appears at least $\Delta = \Delta(n,f)$ times in $G$, where $\Delta$ is a threshold to be optimized later.
    Let $\mathcal{R}$ denote the rest of the colors, not in $\mathcal{H}$,
    that have \emph{low prevalence}.
    At a high level, the failure of any color with high prevalence $c \in \mathcal{H}$ is handled by recursively invoking labels for $f-1$ color faults, but in the graph $G-c$.
    The complementary case, where all failing colors have low prevalence, is treated using the edge labels of \cite{DoryP21}, which are crucially capable of handling any number of individual edge faults.

    We use the following notations.
    Let $G'$ be a subgraph of $G$.
    The function $L_{f-1} (\cdot, G')$ denotes the labels assigned to the vertices and colors of $G'$ by the (recursively defined) labeling scheme for $f-1$ color faults.
    The function $L_{\eft}(\cdot, G')$ denotes the labels assigned to the vertices and edges of $G'$ by the labeling scheme of \Cref{thm:edge-fault-lables}.
    For a color $c$ in $G$, let $E_c \subseteq E$ be the subset of $G$-edges with color $c$.
    
    \medskip
    \highlight{Labeling.}
    The labeling procedure is presented as \Cref{alg:f-fault-labels}.
    The labels $L_1(\cdot, G)$ (the base case $f=1$) are given by \Cref{thm:single-fault-upper-bound}.

    \begin{algorithm}[h]
    \caption{Creating the labels}\label{alg:f-fault-labels}
    % \textbf{Input:} Colored $n$-vertex graph $G$, fault parameter $f \geq 2$ \\
    % \textbf{Output:} Labels $L_f (v, G)$ for each vertex $v \in V$, and $L_f (c, G)$ for each color $c$
    \begin{algorithmic}[1]
    \Require Colored $n$-vertex graph $G$, fault parameter $f \geq 2$
    \Ensure Labels $L_f (v, G)$ for each vertex $v \in V$, and $L_f (c, G)$ for each color $c$
    \State $\Delta \gets \Delta(n,f)$  \Comment{prevalence threshold}
    \State $\mathcal{H} \gets$ set of colors $c$ with $|E_c| \geq \Delta$
    \State $\mathcal{R} \gets$ set of colors $c$ with $|E_c| < \Delta$
    \For{each vertex $v \in V$}
        \Comment{create the label $L_f(v, G)$}
        \State \textbf{store} in $L_f(v,G)$ the labels $L_{f-1}(v, G-h)$ for every $h \in \mathcal{H}$
        \State \textbf{store} in $L_f(v,G)$ the label $L_{\eft}(v, G)$
    \EndFor
    \For{each color $c$}
        \Comment{create the label $L_f(c,G)$}
        \State \textbf{store} in $L_f(c,G)$ the labels $L_{f-1}(c, G-h)$ for every $h \in \mathcal{H}$
        \If{$c \in \mathcal{R}$}
            \State \textbf{store} in $L_f(c,G)$ the labels $L_{\eft}(e, G)$ of every $e \in E_c$
        \EndIf
    \EndFor
    \end{algorithmic}
    \end{algorithm}

    \medskip
    \highlight{Answering queries.}
    Let $u,v \in V$, and let $F$ be a set of at most $f$ colors.
    Given $L(u),L(v)$ and $\{L(c) \mid c \in F\}$, we should determine if $u$ and $v$ are connected in $G-F$.
    There are two cases:

    \begin{enumerate}
        
    \item
    If $F \cap \mathcal{H} \neq \emptyset$:
    Let $h \in F \cap \mathcal{H}$, and denote $F' = F-\{h\}$.
    Note that $|F'| \leq f-1$.
    The $L_{f-1}(\cdot, G-h)$-labels of $u,v$ and of every $c \in F'$ are stored in their respective $L_f(\cdot, G)$-labels.
    By induction, we use these to determine (w.h.p) if $u,v$ are connected in $(G-h)-F' = G-F$.

    \item
    If $F \cap \mathcal{H} = \emptyset$:
    Then $F \subseteq \mathcal{R}$,
    so for every $c \in F$, the labels $L_{\eft}(e, G)$ of every $e \in E_c$ are found in $L_f (c, G)$.
    The $L_{\eft}(\cdot, G)$-labels of $u,v$ are found in their respective $L_f (\cdot, G)$-labels.
    By \Cref{thm:edge-fault-lables},
    using these we can determine (w.h.p) if $u,v$ are connected in $G - \bigcup_{c \in F} E_c = G-F$.
    
    \end{enumerate}

    \medskip
    \highlight{Length analysis.}
    Without loss of generality, we may assume that $G$ has $O(fn \log(fn))$ edges; otherwise, we can replace $G$ by its subgraph given by \Cref{lem:conn-cert}.
    Since $|E(G)| = O(fn \log(fn))$, then
    % This implies that 
    $|\mathcal{H}| = O(\Delta^{-1} fn\log(fn)))$.
    Let $b(n,f-1)$ be a bound on the bit-length of an $L_{f-1}(\cdot, \cdot)$ label assigned for an $n$-vertex graph.
    Define $b(n,f)$ similarly for $L_f(\cdot, \cdot)$.
    The largest $L_f(\cdot,G)$-labels are given for colors from $\mathcal{R}$: they store $|\mathcal{H}|$ of the $L_{f-1}(\cdot, \cdot)$-labels and $\Delta$ of the $L_{\eft}(\cdot,\cdot)$ labels.
    This gives the following recursion:
    \[
    b(n,f) = O\big( b(n,f-1) \cdot \Delta^{-1}fn\log(fn) + \paren{\log^3 n + \log(fn)} \cdot \Delta \big).
    \]
    To minimize the sum, we set $\Delta$ to make both terms equal, so that:
    \[
    \Delta = \Delta(n,f) = \sqrt{fn \log(fn) b(n,f-1) / (\log^3 n + \log(fn))},
    \]
    % which results in
    \[
    b(n,f) = O\paren{\sqrt{b(n,f-1) fn \log(fn) \paren{\log^3 n + \log(fn)}}}.
    \]
    Solving this recursion, with base case $b(n,1) = O(\sqrt{n} \log n)$ given by \Cref{thm:single-fault-upper-bound}, yields
    \[
    b(n,f) = O\paren{\sparen{fn \log(fn) (\log^3 n + \log(fn))}^{1 - 1/2^f} }.
    \]
    % This concludes the proof.
    This concludes the proof of \Cref{lem:small-f-upper-bound} for the edge-colored case.
% \end{proof}
The proof extends easily to vertex-colored graphs, by classifying a color $c$ as having high prevalence if the \emph{volume} (sum of degrees) of the set of vertices with color $c$ is above the threshold $\Delta$.

\subsubsection{Labeling Scheme for $f = \Omega(\log \log n)$}

\begin{lemma}\label{lem:large-f-upper-bound}
    There is a randomized labeling scheme for connectivity under $f$ color faults, assigning labels of length ${O}(n \cdot \polylog(fn))$ bits on $n$-vertex graphs.
\end{lemma}
\begin{proof}
    Let $G = (V,E)$ be the input colored $n$-vertex graph with colors from $[C]$.
    By \Cref{lem:conn-cert}, we may assume that $G$ has $\tilde{O}(fn)$ edges.
    For a color $c$, denote by $E_c \subseteq E$ the set of edges with color $c$,
    and let $T_c$ be a spanning forest of the subgraph $(V,E_c)$.
    Finally, let $H = \bigcup_{c\in [C]} T_c$. 
    We show that for all $F \subseteq [C]$, every pair of vertices are connected in $H-F$ iff they are connected in $G-F$.
    It suffices to prove that if $e = \{u,v\}$ is an edge of $G-F$, then there is some $u$-$v$ path in $H-F$.
    If $e$ has color $c \notin F$, then
    by construction of $T_c$, there is a $c$-colored path in this forest connecting $u$ and $v$,
    which is also present in $H-F$.

    \medskip
    \highlight{Labeling.}
    Apply the Dory-Parter \cite{DoryP21} scheme of \Cref{thm:edge-fault-lables} on $H$, resulting in labels for the vertices and edges of $H$, denoted $L_{\eft}(\cdot, H)$.
    The label $L(v)$ of a vertex $v \in V$ simply stores $L_{\eft}(v, H)$.
    The label $L(c)$ of a color $c$ stores $L_{\eft}(e, H)$ of \emph{every} edge $e \in E(T_c)$.
    The claimed length bound is immediate, as storing a single $L_{\eft}(\cdot, H)$-label requires $O(\log^3 n + \log(fn))$ bits, and $|E(T_c)| \leq n-1$.

    \medskip
    \highlight{Answering queries.}
    Let $u,v \in V$ and $F \subseteq [C]$.
    The labels $L(u)$, $L(v)$ and $\{L(c) \mid c \in F\}$,
    stores the $L_{\eft}(\cdot, H)$-labels of $u,v$ and every $e \in \bigcup_{c\in F} E(T_c)$.
    Using these, we can, with high probability, determine the connectivity of $u,v$ in $H - \bigcup_{c\in F} E(T_c) = H-F$, and hence also in $G-F$, with high probability.

    \medskip
    The proof extends to vertex colors, where $E_c$ is the set of edges that touch the color $c$.
\end{proof}

\Cref{thm:f-faults-upper-bound} follows by combining \Cref{lem:small-f-upper-bound} and \Cref{lem:large-f-upper-bound}.

\subsection{Lower Bound}\label{sect:f-faults-lower-bound}

We next provide a lower bound that generalizes the $\Omega(\sqrt{n})$-bit lower bound for the case $f=1$ of \Cref{thm:single-fault-lower-bound}.
However, in contrast to \Cref{thm:single-fault-lower-bound}, this lower bound is \emph{existential}, namely, it relies on some specific `worst-case' topology.

\begin{theorem}\label{thm:f-faults-lower-bound}
Let $f \geq 1$ be a fixed constant.
Every (possibly randomized) labeling scheme for connectivity under $f$ color faults in $n$-vertex graphs must have label length of $\Omega(n^{1-1/(f+1)})$ bits.
Furthermore, this bound holds even for labeling schemes restricted to simple planar graphs.
\end{theorem}

\begin{proof}
    Suppose there is such a labeling scheme with label length of $b$ bits.
    The proof strategy is similar to the proof of \Cref{thm:single-fault-lower-bound}:
    using the labeling scheme to devise a one-way communication protocol for the indexing problem $\INDEX(N)$, with $N  = \Theta(n)$.
    Let $x = x_0 x_1 \cdots x_{N-1}$ be the input string of Alice, and $i^* \in [0, N)$ be the input index of Bob.

    The communication protocol relies on a specific (uncolored) $n$-vertex graph topology $G$, known in advance to Alice and Bob, which we now define.
    First, denote
    \[
    M \bydef \binom{N^{1/(f+1)}}{f} = \Theta( N^{1 - 1/(f+1)}).
    \]
    The topology $G$ is an ``$f$-thick spider'' with $N/M = \Theta(N^{1/(f+1)})$ arms, each of length $M$.
    Formally, it consists of a starting vertex $s$, from which there are $N/M$ emanating ``$f$-thick'' paths $P_0, P_1, \dots, P_{N/M-1}$, where two consecutive vertices in a path have $f$ parallel edges between them.
    Each such path $P_k$ consists of $M+1$ vertices, and is disjoint from the other paths except for the common starting vertex $s$.
    We denote the vertex of distance $l$ from $s$ in the path $P_k$ by $v_{k,l}$, so $s = v_{k,0}$.
    We also use the shorthand notation $t_k = v_{k,M}$ for the last vertex of path $P_k$.
    The set of $f$ parallel edges between $v_{k,l}$ and $v_{k,l+1}$ is denoted by $E_{k,l}$.
    See \Cref{fig:f-fault-lower-bound}(Left) for an illustration.
    \begin{figure}
    \centering
    \includegraphics[width=1.0\textwidth]{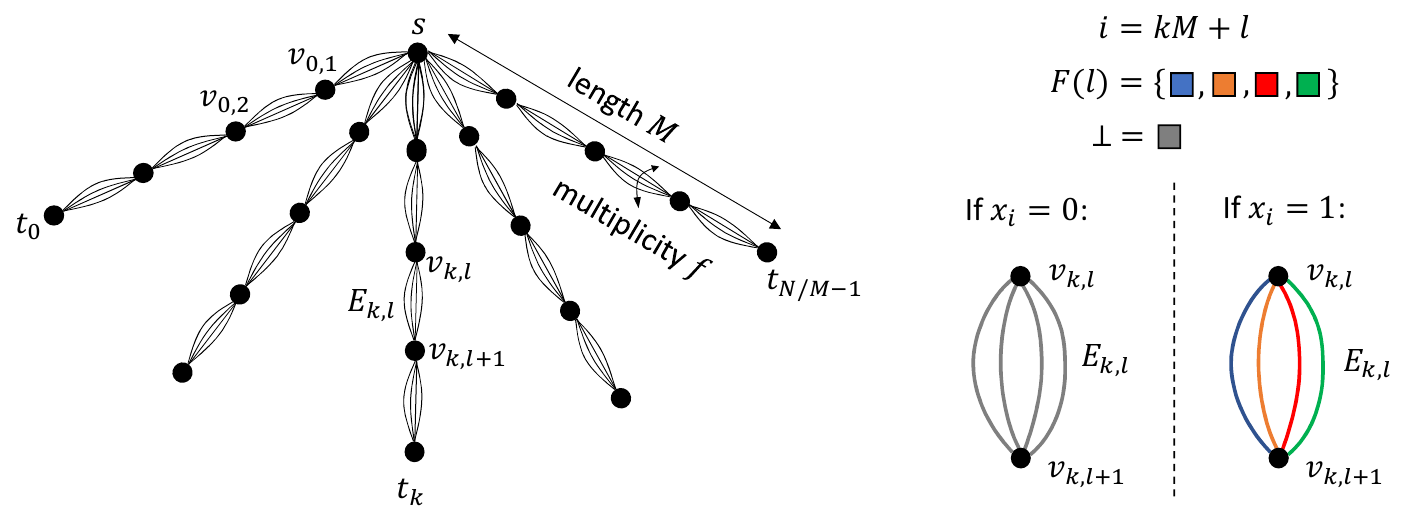}
    \caption{
    Illustration of the proof of \Cref{thm:f-faults-lower-bound}.
    Left: The ``$f$-thick spider'' graph topology $G$.
    Right: The coloring procedure for $E_{k,l}$.}
    \label{fig:f-fault-lower-bound}
    \end{figure}

    Consider 
    the color-palette %to be used by Alice is 
    $\{0, 1, \dots, N^{1/(f+1)}-1\} \cup \{\perp\}$, where $\perp$ is a never-failing null color.
    By its definition, $M$ is the number of $f$-subsets of the color-set.
    Fix a bijection $F$ mapping $l \in [0, M)$ to a unique $f$-subset of colors $F(l) \subseteq \{0, \dots, N^{1/(f+1)}-1\}$.
    (The bijection $F$ is also part of the protocol, i.e., known in advance to Alice and Bob.)

    We are now ready to describe the protocol.
    Alice colors her copy of $G$ according to her input $x$, as follows.
    For each $k \in [0,N/M)$ and $l \in [0,M)$, Alice considers the corresponding index $i = kM+l$.
    If $x_i = 1$, she colors each of the $f$ edges in $E_{k,l}$ with a distinct color from the $f$ colors in $F(l)$.
    Otherwise, when $x_i = 0$, she colors all of $E_{k,l}$ with the null color $\perp$.
    See illustration in 
    \Cref{fig:f-fault-lower-bound}(Right).

    This coloring procedure ensures the following property.
    Let $k$, $l$ and $i$ as before, and consider $P_k-F(l)$ (i.e., path $P_k$ after all $f$ colors in $F(l)$ fail). For $l' \neq l$, the edge-set $E_{k,l'}$ has at least one surviving edge: it is either entirely colored with the non-failing $\perp$, or it contains all $f$ colors of $F(l')$, of which at least one is non-faulty as $F(l') \neq F(l)$.
    Therefore, $P_k-F(l)$ is connected iff $v_{k,l}$ connects to $v_{k,l+1}$ after all $f$ colors in $F(l)$ fail, which happen iff $x_i=0$.

    Next, Alice assigns labels using the assumed labeling scheme, and sends to Bob the labels of the vertices $s$ and $t_0, \dots, t_{N/M-1}$, and the labels of the colors $0,1, \dots, N^{1/(f+1)}-1$.
    To recover $x_{i^*}$,
    Bob decomposes $i^*$ as $i^* = k^* \cdot M + l^*$
    with $k^* \in [0,N/M)$, $l^* \in [0,M)$,
    and uses the received labels of $s, t_{k^*}$ and the $f$ colors in $F(l^*)$ to query the connectivity of $s$ and $t_{k^*}$ in $G - F(l^*)$.
    If the answer is \emph{disconnected}, Bob determines that $x_{i^*} = 1$, and if it is \emph{connected}, he determines that $x_{i^*} = 0$.
    This establishes the protocol.
    The correctness is guaranteed by the previously described property of the coloring.
    The total number of bits sent by Alice is $O(b N^{1/(f+1)})$. Thus, by \Cref{lem:index_LB}, $O(b N^{1/(f+1)}) = \Omega(N)$, and hence $b = \Omega(N^{1-1/(f+1)}) = \Omega(n^{1-1/(f+1)})$.

    Finally, for the `furthermore' part, we can alter $G$ by subdividing the edges, i.e., replacing each edge with a length-two path which is colored according to the original color of its corresponding edge in the $f$-thick spider.
    The resulting graph is simple and planer, and the number of vertices only increases by a factor of $O(f) = O(1)$, so it remains $\Theta(N)$, and the proof goes through.

    The proof for vertex-colored graphs also follows by subdividing the edges in the above manner, where new vertices get the color of their corresponding original edge, and original vertices get the color $\perp$.
\end{proof}

We note that in some sense, the proof technique of~\Cref{thm:f-faults-lower-bound} cannot be used to obtain a lower bound stronger than $\tilde{\Omega}(n^{1-1/(f+1)})$.
See \Cref{par:f_faults_LB_barrier} for a further discussion on this limitation.

\section{Two Color Faults in Small-Diameter Graphs}\label{sect:two-fault-diameter}

In this section, we provide a nearly optimal labeling scheme under two color faults for graphs with diameter $D=\tilde{O}(1)$.

\subsection{Upper Bound}

% \begin{theorem}\label{thm:two-fault-upper-bound}
%     There is a deterministic labeling scheme for connectivity with two color faults that, when given an $n$-vertex graph $G$ with diameter $D$, assigns labels of length $O(D(\sqrt{n}+D)\log^2 n)$ bits.

%     Moreover, the result extends
%     to graphs with several connected components -- same bound, but with $D$ being the maximal diameter of a connected component in $G$.
% \end{theorem}

\begin{theorem}\label{thm:two-fault-upper-bound}
    There is a deterministic labeling scheme for connectivity with two color faults that, when given an $n$-vertex graph $G$ with all connected components of diameter at most $D$, assigns labels of length $O(D(\sqrt{n}+D)\log^2 n)$ bits.
\end{theorem}

\begin{proof}
We focus on the case where $G$ is connected, and later mention the straightforward modifications to obtain the general case.
We use the following notation.
For a color $c \in [C]$ and a subgraph $H$ of $G$, the notation $c \in H$ means that $c$ appears in $H$.
When $H$ is a tree, and $u,v \in V(H)$, the unique $u$-$v$ path in $H$ is denoted by $H[u,v]$.

\medskip
\highlight{Preprocessing.}
Our labeling relies on several trees formed by executing breadth-first search procedures (in short, BFS trees), which we now define.
First, let $T$ be a BFS tree in $G$, rooted at the vertex $s$ with minimum $\id$.
Next, for each $v \in V$ and color $c \in T[s,v]$, we let $T_{v,c}$ be the tree formed by executing a BFS procedure from $v$ in $G-c$, but halting once $\sqrt{n}$ vertices are reached.
Note that if $T_{v,c}$ contains fewer than $\sqrt{n}$ vertices, then it is a spanning tree for the connected component of $v$ in $G-c$.
However, if $T_{v,c}$ has $\sqrt{n}$ vertices, it might not span this entire component.
We next compute a \emph{hitting set} $U \subseteq V$ for the trees $T_{v,c}$ that have $\sqrt{n}$ vertices.
That is, for every $v \in V$ and $c \in T[s,v]$, if $T_{v,c}$ contains $\sqrt{n}$ vertices, then it contains some vertex $u_{v,c} \in U$.
It is well known that
such a hitting set $U$ with $|U| = O(\sqrt{n} \log n)$ can be constructed efficiently (e.g., using the greedy algorithm).
For completeness, we include a proof in \Cref{sect:hitting-set} (\Cref{lem:det_hitting_set}).

\medskip
\highlight{Labeling.}
The label $L(v)$ of a vertex $v \in V$ is constructed by \Cref{alg:two-fault-vertex-labels}.
The label $L(c)$ of a color $c \in [C]$ is constructed by \Cref{alg:two-fault-color-labels}.

\begin{algorithm}[h]
\caption{Creating the label $L(v)$ of vertex $v \in V$}\label{alg:two-fault-vertex-labels}
\begin{algorithmic}[1]
    \For{each color $c \in T[s,v]$}
        \State \textbf{store} $\cid(v, G-c)$
        \State \textbf{store} $\cid(v,G-\{c,d\})$ for every color $d \in T_{v,c}$
        \If{$T_{v,c}$ has $\sqrt{n}$ vertices}
            \State $u_{v,c} \gets$ a vertex from $U$ present in $T_{v,c}$
            \State \textbf{store} $\id(u_{v,c})$
            \State \textbf{store} $\cid(u_{v,c},G-\{c,d\})$ for every $d \in T[s,u_{v,c}]$
        \EndIf
    \EndFor
\end{algorithmic}
\end{algorithm}
\setlength{\intextsep}{0pt}% Remove \textfloatsep
\begin{algorithm}[h]
\caption{Creating the label $L(c)$ of color $c \in [C]$}\label{alg:two-fault-color-labels}
\begin{algorithmic}[1]
    \For{each vertex $u \in U$}
        \For{each color $d \in T[s,u]$}
            \State \textbf{store} $\cid(u, G-\{c,d\})$
        \EndFor
    \EndFor
\end{algorithmic}
\end{algorithm}

\medskip
\highlight{Length analysis.}
Recall that $T$ is a BFS tree for $G$, so its depth is at most $D$.
As $|U| = O(\sqrt{n} \log n)$, we obtain that a color label $L(c)$ stores $O(D\sqrt{n} \log^2 n)$ bits.
Consider now a vertex label $L(v)$.
Note that $T[s,v]$ has at most $D$ edges. For every $c \in T[s,v]$, we have that $T_{v,c}$ has at most $\sqrt{n}$ edges, and $T[s,u_{v,c}]$ (when defined) has at most $D$ edges.
Therefore, the label $L(v)$ stores only $O(D \cdot (\sqrt{n} + D))$ $\cid$s, which requires $O(D(\sqrt{n}+D)\log n)$ bits.
In total, all labels store $O(D(\sqrt{n}+D)\log^2 n)$ bits.

\medskip
\highlight{Answering queries.}
Given the labels $L(v), L(c), L(d)$ of $v \in V$ and two colors $c,d \in [C]$, we show a procedure for deducing $\cid(v,G-\{c,d\})$.

If both $c$ and $d$ do not appear on $T[s,v]$, then $\cid(v,G-\{c,d\}) = \cid(s, G-\{c,d\}) = \id(s)$, where the last equality is by choice of $s$ as the vertex with minimum $\id$,
and we are done (as the minimum $\id$ can be assumed to be fixed, say to $1$).
From now on, assume that one of the failing colors, say $c$, appears on $T[s,v]$.
If $d \in T_{v,c}$, then $\cid(v, G-\{c,d\})$ is found in $L(v)$, and we are done.
So, assume further that $d \notin T_{v,c}$.

We now treat the case where $T_{v,c}$ has fewer than $\sqrt{n}$ vertices.
Then $T_{v,c}$ spans the connected component of $v$ in $G-c$.
As $d \notin T_{v,c}$, it must be that this is also the connected component of $v$ in $G-\{c,d\}$.
Therefore, $\cid(v,G-\{c,d\})=\cid(v,G-c)$, and the latter is stored in $L(v)$, so we are done.

Next, we handle the case where $T_{v,c}$ has $\sqrt{n}$ vertices, so $u_{v,c}$ is defined.
As $d \notin T_{v,c}$, and also $c \notin T_{v,c}$ (by definition),
the path $T_{v,c}[v,u_{v,c}]$ connects $v$ and $u_{v,c}$ in $G-\{c,d\}$, implying that $\cid(v, G-\{c,d\}) = \cid(u_{v,c}, G-\{c,d\})$ and therefore it is enough to show how to find  $\cid(u_{v,c}, G-\{c,d\})$.
There are three options:
\begin{enumerate}
    \item If $c,d \notin T[s,u_{v,c}]$, then $\cid(u_{v,c}, G-\{c,d\}) = \cid(s, G-\{c,d\}) = \id(s)$.
    \item If $c \in T[s,u_{v,c}]$, then $L(d)$ stores $\cid(u_{v,c}, G-\{c,d\})$.
    \item If $d \in T[s,u_{v,c}]$, then $L(v)$ stores $\cid(u_{v,c}, G-\{c,d\})$.
\end{enumerate}
As at least one of these options must hold, this concludes the proof for connected graphs.

To handle graphs with several connected components, the proof is modified by replacing the single BFS tree $T$ with a collection of BFS trees, one for each connected component of $G$.

The proof works as-is also for vertex-colored graphs.
\end{proof}

% \begin{remark}
%     It is straightforward to strengthen \Cref{thm:two-fault-upper-bound} to handle graphs with several connected components (and therefore infinite diameter).
%     In this case, we get the same bound, but with $D$ being the maximal diameter of a connected component in $G$.
%     The proof is modified by replacing the single BFS tree $T$ with a collection of BFS trees, one for each connected component of $G$.
% \end{remark}

\subsection{Lower Bounds}

We now show that \Cref{thm:two-fault-upper-bound} is optimal for graphs with $D=\Tilde{O}(1)$.

\begin{theorem}\label{thm:two-fault-lower-bound}
    Let $G$ be a graph topology and $H$ be a subgraph of $G$. Then, every labeling scheme for connectivity under two color faults in (colorings of) $G$ must have label length of $\Omega(\bp(H))$ bits.
\end{theorem}

\begin{proof}
    The proof is by a reduction. We use connectivity labels for \emph{two} color faults in $G$, denoted $L_{G,2}(\cdot)$, to construct such labels for \emph{one} color fault in $H$, denoted $L_{H,1}(\cdot)$, as follows.
    Given a coloring of $H$ with palette $[C]$, we extend it to a coloring of $G$ by assigning to all edges of $G$ that are not in $H$ a new fixed color $c' \notin [C]$.
    The label $L_{H,1} (v)$ of each vertex $v \in V(H)$ simply stores $L_{G,2}(v)$.
    The label of $L_{H,1} (c)$ of a color $c \in [C]$ stores the pair $\ang{L_{G,2}(c), L_{G,2}(c')}$.
    Thus, given the $L_{H,1}(\cdot)$ labels of $u,v \in V(H)$ and $c \in [C]$, one can use the $L_{G,2}(\cdot)$ labels stored in them to decide if $u,v$ are connected in $G-\{c,c'\}$, which happens iff $u,v$ are connected in $H-c$.
    By \Cref{thm:single-fault-lower-bound}, the $L_{H,1}(\cdot)$ labels must have maximum length $\Omega(\bp(H))$, which implies the same conclusion for the $L_{G,2}(\cdot)$ labels.
\end{proof}

% Since a cycle of length $k$ has ball packing number $\Theta(\sqrt{k})$, we obtain:

% \begin{corollary}
%     Every connectivity labeling scheme under two color faults for an Hamiltonian graph $G$ must have labels of  $\Omega(\sqrt{n})$ bits.
% \end{corollary}

% Since the wheel graph on $n$ vertices $W_n$ is Hamiltonian and has $\diam(W_n)=2$, we have:

\begin{corollary}
    For every $n$, there exists a graph $G$ on $n$ vertices with $\diam(G)=2$, for which every connectivity labeling scheme under two color faults must have labels of length $\Omega(\sqrt{n})$.
\end{corollary}
\begin{proof}
    Let $G$ be the wheel graph on $n$ vertices, composed of a cycle $C$ on $n-1$ vertices, and another vertex with an edge going to each of them.
    The result follows from \Cref{thm:two-fault-lower-bound} as $\diam(G)=2$, and $\bp(C) = \Omega(\sqrt{n})$.
\end{proof}

\section{Conclusion}\label{sect:conclusion}
In this work, we introduced $f$ color fault-tolerant connectivity labeling schemes, which generalize the well-studied edge/vertex fault-tolerant connectivity labeling schemes.
Our results settle the complexity of the problem when $f=1$.
For $f\geq 2$, many interesting open problems remain:
\begin{itemize}
    \item Can we close the gap between the $\tilde{O}(n^{1-1/2^f})$ and $\Omega(n^{1-1/(f+1)})$ bounds?
    Concretely, is there a labeling scheme for connectivity under $f=2$ color faults with labels of $\tilde{O}(n^{2/3})$ bits?
    Can our solution for low-diameter graphs be utilized to obtain such a scheme?   

    \item Is there a graph parameter that generalizes $\bp(G)$ and characterizes the length of a universally optimal labeling scheme for $f \geq 2$? 
    Notably, the proof of
    % this seems to be an intriguing question as
    \Cref{thm:two-fault-lower-bound} implies that even very simple graphs with small diameter and ball packing number admit a lower bound of $\Omega(\sqrt{n})$ bits for $f=2$.

    \item Can we provide non-trivial \emph{centralized oracles} for connectivity under $f \geq 2$ color faults?

    \item 
    Are there routing schemes for avoiding $f\geq 2$ forbidden colors with small header size?
    Our labeling scheme for $f\geq 2$ could be extended to such a routing scheme, but with a large header size of $\tilde{O}(n^{1-1/2^f})$ bits.
\end{itemize}

Another intriguing direction is going beyond connectivity queries; a natural goal is to additionally obtain approximate distances, which is open even for $f=1$.
This problem is closely related to providing forbidden color routing schemes with good stretch guarantees.

\ifDRAFT
    % no acknowledgments
\else
\paragraph{Acknowledgments.}
We are grateful to Merav Parter for encouraging this collaboration, and for helpful guidance and discussions.
\fi

% \cleardoublepage
\phantomsection
\addcontentsline{toc}{section}{References}

\ifPODCSUB
    \bibliographystyle{ACM-Reference-Format}
\else
    \bibliographystyle{alphaurl}
\fi
\bibliography{references.bib}

% \newpage
\appendix

% \section{Reduction to the Single-Source Variant}
\section{Reduction from All-Pairs to Single-Source}
\label{sect:single-source-reduction}

In the \emph{single-source} variant of fault-tolerant connectivity, the input graph $G$ comes with a designated \emph{source vertex} $s$.
The queries to be supported are of the form $\ang{u,F}$, where $u\in V$ and $F$ is a faulty set of size at most $f$.
It is required to report if $u$ is connected to the source $s$ in $G-F$.
The following result shows that this variant is equivalent to the all-pairs variant, up to $\log n$ factors.
The result holds whether the faults are edges, vertices, or colors, hence we do not specify the type of faults.

% Consider the \emph{single-source} variant of fault-tolerant connectivity labeling: Given an $n$-vertex graph $G$ with a designated \emph{source vertex} $s$, and a bound $f$ on the number of faults, one should assign short labels such that for every vertex $v$ and fault-set $F$ with $|F|\leq f$, the connectivity of $s$ and $v$ in $G-F$ can be determined from the labels of $v$ and $F$.
% Here, and throughout this section, we do not care about the type of faulty elements; these could be edges, vertices or colors.

% Clearly, every labeling scheme for \emph{all-pairs} fault-tolerant connectivity can be transformed into a single-source variant, by including the label of $s$ in all other labels, which at most doubles the label length.
% In this section, we consider the converse direction, and show that a single-source scheme can be used in a black-box manner to obtain an all-pairs scheme with only a small overhead in length.

\begin{theorem}
    Let $f \geq 1$.
    Suppose there is a (possibly randomized) single-source $f$ fault-tolerant connectivity labeling scheme that assigns labels of at most $b(n, f)$ bits on every $n$-vertex graph.
    Then, there is a randomized all-pairs $f$ fault-tolerant connectivity labeling scheme that assigns labels of length $O(b(n+1, f) \cdot \log^2 n)$ bits on every $n$-vertex graph.
\end{theorem}

\begin{proof}
    \highlight{Labeling.}
    For each $i,j$ with $1 \leq i \leq \ceil{\alpha \ln(n) / \ln(0.9)}$, $1 \leq j \leq \ceil{\log_2 n} + 2$ we independently construct a graph $G_{ij}$ as follows:
    Start with $G$, add a new vertex $s_{ij}$, and independently for each $v \in V$, add a new edge connecting $s_{ij}$ to $v$ with probability $2^{-j}$. 
    The vertex $s_{ij}$ and the new edges are treated as \emph{non-failing}.
    That is, in case of color faults, they get a null-color $\perp$ that does not appear in $G$.
    For each element (vertex/edge/color) $x$ of $G$, its label $L(x)$ is the concatenation of all $L_{ij}(x)$, where the $L_{ij}(\cdot)$ are the labels given by the single-source scheme to the instance $G_{ij}$ with designated source $s_{ij}$.
    The claimed length bound is immediate.

    \medskip
    \highlight{Answering queries.}
    Let $u,w \in V$, and let $F$ be a fault-set of size at most $f$.
    Given $L(u), L(w)$ and $\{L(x) \mid x \in F\}$, we should determine if $u,w$ are connected in $G-F$.
    To this end, for each $i,j$, we use the $L_{ij}(\cdot)$ labels of $u,F$ to determine if $u$ is connected to $s_{ij}$ in $G_{ij}$, and do the same with $w$ instead of $u$.
    If the answers are always identical for $u$ and $w$, we output \emph{connected}. Otherwise, we output \emph{disconnected}.
    
    \medskip
    \highlight{Analysis.}
    We have made only $O(\log^2 n)$ queries using the single-source scheme, so, with high probability, all of these are answered correctly. Assume this from now on.

    If $u,w$ are connected in $G-F$, then this is also true for all $G_{ij}-F$, so they must agree on the connectivity to $s_{ij}$ in this graph.
    Hence, in this case, the answers for $u$ and $w$ are always identical, and we correctly output \emph{connected}.

    Suppose now that $u$ and $w$ are disconnected in $G-F$.
    Let $U$ be the set of vertices in $u$'s connected component in $G-F$.
    Define $W$ analogously for $w$.
    Without loss of generality, assume $|U| \leq |W|$.
    Let $j$ be such that $2^{j-2} < |U| \leq 2^{j-1}$.
    Let $N^{(i)}_U$ be the number of edges between $s_{ij}$ and $U$ in $G_{ij}$, and define $N^{(i)}_W$ similarly.
    % Note that $\mathbb{E}[N^{(i)}_U] = |U|/ 2^j \leq 1/2$, so by Markov's inequality, $\Pr[N^{(i)}_U = \emptyset] \geq 1/2$.
    % On the other hand, 
    % \[
    % \Pr[N^{(i)}_W \neq \emptyset] = 1 - \paren{1 - \frac{1}{2^j}}^{|W|} \geq 1 - \paren{1 - \frac{1}{2^j}}^{2^{j-2}}
    % \geq 1 - e^{-1/4} > 1/5.
    % \]
    By Markov's inequality,
    \begin{align*}
        \Pr\big[N^{(i)}_U = 0\big] &\geq 1 - \mathbb{E}\big[N^{(i)}_U\big] = 1 - |U| \cdot 2^{-j} \geq 1 - 2^{j-1} \cdot 2^{-j} = 1/2.\\
        \shortintertext{On the other hand,}
        \Pr[N^{(i)}_W \geq 1] &= 1 - \paren{1 - 2^{-j}}^{|W|} \geq 1 - \paren{1 - 2^{-j}}^{2^{j-2}}
    \geq 1 - e^{-1/4} > 0.2.
    \end{align*}
    % \[
    % \Pr\big[N^{(i)}_U = 0\big] \geq 1 - \mathbb{E}\big[N^{(i)}_U\big] = 1 - |U| \cdot 2^{-j} \geq 1 - 2^{j-1} \cdot 2^{-j} = 1/2.
    % \]
    % On the other hand,
    % \[
    % \Pr[N^{(i)}_W \geq 1] = 1 - \paren{1 - 2^{-j}}^{|W|} \geq 1 - \paren{1 - 2^{-j}}^{2^{j-2}}
    % \geq 1 - e^{-1/4} > 0.2.
    % \]
    Since $U$ and $W$ are disjoint, $N^{(i)}_U$ and $N^{(i)}_W$ are independent random variables.
    Hence, with probability at least $0.1$,
    the source $s_{ij}$ is connected to $w$ but not to $u$ in $G_{ij} - F$, and the answers for $u$ and $w$ given by the $L_{ij}(\cdot)$-labels are different.
    As the graphs $\{G_{ij}\}_i$ are formed independently, 
    the probability there exists an $i$ for which $w$ is connected to $s_{ij}$ and $u$ is disconnected from $s_{ij}$
    is at least $1 - (0.9)^{\alpha \ln n / \ln(0.9)} = 1 - 1/n^\alpha$.
    In this case, the output is \emph{disconnected}, as required.
\end{proof}

\section{Routing }\label{sec:routing}

In this section, we provide a routing scheme for avoiding any single forbidden color.
This is a natural extension of the forbidden-set routing framework, initially introduced by \cite{CourcelleT07} (see also \cite{AbrahamCG12,AbrahamCGP16,DoryP21,ParterPP23}), to the setting of colored graphs.
We refer the reader to \cite{DoryP21} for an overview of forbidden-set routing, and related settings.
% as well as the harder setting of fault-tolerant routing, where the faults are unknown to the source.
Such a routing scheme consists of two algorithms.
The first is a preprocessing (centralized) algorithm that computes \emph{routing tables} to be stored at each vertex of $G$, and \emph{labels} for the vertices and the colors.
The second is a distributed routing algorithm that enables routing a message $M$ from a source vertex $s$ to a target vertex $t$ avoiding edges of color $c$.
Initially, the labels of $s,t,c$ are found in the source $s$. Then, at each intermediate node $v$ in the route, $v$ should use the information in its table, and in the (short) header of the message, in order to determine where the message should be sent;
formally, $v$ should compute the \emph{port number} of the next edge to be taken from $v$ (which must not be of color $c$).
It may also edit the header for future purposes.

The main concern is minimizing the size of the tables and labels, and even more so of the header (as it is communicated through the route).
Another important concern is optimizing the \emph{stretch}, which is the ratio between the length of the routing path and the length of the shortest $s,t$ path in $G-c$.
Unfortunately, our routing scheme does not provide good stretch guarantees, and optimizing the stretch is an interesting direction for future work.
We note, however, that the need to avoid edges of color $c$ by itself poses a nontrivial challenge, and black-box application of the state-of-the-art routings schemes of Dory and Parter~\cite{DoryP21} for avoiding $f = \Omega(n)$ individual edges would yield large labels, tables and headers, \emph{and} large stretch (all become $\Omega(n)$).
We show:

\begin{theorem}\label{thm:routing}
    There is a deterministic routing scheme for avoiding one forbidden color such that, for a given colored $n$-vertex graph $G$, the following hold:
    \begin{itemize}
        \item The routing tables stored at the vertices are all of size $O(\bp(G) \log n)$ bits.
        \item The labels assigned to the vertices and the colors are of size $O(\bp(G) \log n)$ bits.
        \item The header size required for routing a message contains only $O(\log n)$ bits.
    \end{itemize}
 \end{theorem}

The rest of the section is devoted to proving the above theorem.
For the sake of simplicity, we assume that when $c$ is the color to be avoided, the graph $G-c$ is connected. (In particular, this also implies that $G$ is connected.)
Intuitively, this assumption is reasonable as we cannot route between different connected components of $G-c$. To check if the routing is even possible (i.e., if $s$ and $t$ are in the same connected component), we can use the connectivity labels of \Cref{thm:single-fault-upper-bound} at the beginning of the procedure.
Technically, this assumption can be easily removed, at the cost of introducing some additional clutter.

\subsection{Basic Tools}
In this section we provide several basic building blocks on which our scheme is based.

\paragraph{Tree Routing.}
The first required tool is the \emph{Thorup-Zwick tree routing scheme}~\cite{TZ01-b}, which we use in a black-box manner. Its properties are summarized in the following lemma:

\begin{lemma}[Tree Routing~\cite{TZ01-b}]\label{lem:tree-routing}
    Let $T$ be an $n$-vertex tree.
    One can assign each vertex $v \in V(T)$ a \emph{routing table} $R_T (v)$ and a \emph{destination label} $L_T (v)$ with respect to the tree $T$, both of $O(\log n)$ bits.
    For any two vertices $u,v \in V(T)$, given $R_T (u)$ and $L_T (v)$, one can find the port number of the $T$-edge from $u$ that heads in the direction of $v$ in $T$.
\end{lemma}

\paragraph{The Vertex Set $A$.}
Our scheme crucially uses the existence of the set $A$ constructed in the labeling procedure of \Cref{sect:1-fault-main-upper}.
The following lemma summarizes its two crucial properties:
\begin{lemma}\label{lem:A-set}
    There is a vertex set $A \subseteq V$ such that $|A| = O(\bp(G))$, and every vertex $v \in V$ has $\dist_G (v,A) = O(\bp(G))$.
\end{lemma}

\paragraph{Spanning Tree and Recovery Trees.}
We next define several trees that are crucial for our scheme.
First, we construct a specific spanning tree $T$ of $G$, designed so that the $V$-to-$A$ shortest paths in $G$ are tree paths in $T$.
Recall that for every $v\in V$, $P(v)$ a shortest path connecting $v$ to $A$, and $a(v)$ is the $A$-endpoint of this path (see the beginning of \Cref{sect:1-fault-main-upper}).
We choose the paths $P(v)$ consistently, so that if vertex $u$ appears on $P(v)$, then $P(u)$ is a subpath of $P(v)$.
This ensures that the union of the paths $\bigcup_{v\in V} P(v)$ is a forest.
The tree $T$ is created by connecting the parts of this forest by arbitrary edges.
We root $T$ at an arbitrary vertex $r$.

After the failure of color $c$, the tree $T$ breaks into \emph{fragments} (the connected components of $T-c$).
We define the \emph{recovery tree} of color $c$, denoted $T_c$, as a spanning tree of $G-c$ obtained by connecting the fragments of $T-c$ via additional edges of $G-c$. These edges are called the \emph{recovery edges} of $T_c$, and the fragments of $T-c$ are also called fragments of $T_c$.

\paragraph{First Recovery Edges.}
For $u,v\in V$ and color $c$, we denote $e(u,v,c)$ as the first recovery edge appearing in the $u$-to-$v$ path in $T_c$ (when such exists).
Note that we treat this path as directed from $u$ to $v$.
Accordingly, we think of $e(u,v,c)$ as a \emph{directed} edge $(x,y)$ where its first vertex $x$ is closer to $u$, and its second vertex $y$ is closer to $v$.
Thus, $e(u,v,c)$ and $e(v,u,c)$ refer to the same edge, but in opposite directions.
We will use a basic data block denoted $\FirstRecEdge(u,v,c)$ storing the following information regarding $e(u,v,c)$:
\begin{itemize}
    \item The port number of $e(u,v,c)$, from its first vertex $x$ to its second vertex $y$.
    \item The tree-routing label w.r.t.\ $T$ of the first vertex $x$, i.e. $L_T (x)$.
    \item A Boolean indicating whether the second vertex $y$ and $v$ lie in the same fragment of $T-c$.
\end{itemize}
Note that $\FirstRecEdge(u,v,c)$ consists of $O(\log n)$ bits.

\paragraph{$A$-fragments and $B$-fragments.}
A fragment of $T-c$ containing at least one vertex from $A$ is called an \emph{$A$-fragment}.
For convenience, the fragments that are not $A$-fragments (i.e., do not contain $A$-vertices) are called \emph{$B$-fragments}.
Our construction of $T$ ensures the following property:
\begin{lemma}\label{lem:fragments}
    For every color $c$, if vertex $v \in V$ is in a $B$-fragment of $T-c$, then $c\in P(v)$, i.e., the color $c$ appears on the path $P(v)$.
\end{lemma}
\begin{proof}
    By construction, the path $P(v)$ is a tree path in $T$ connecting $v$ to some $a\in A$.
    As $v$ is in a $B$-fragment of $T-c$, this path cannot survive in $T-c$, hence the color $c$ appears on it.
\end{proof}

% Finally, we define the \emph{recovery tree} of color $c$, denoted $T_c$, as a spanning tree of $G-c$ obtained by connecting the fragments of $T-c$ via additional edges of $G-c$. These are called the \emph{recovery edges} of $T_c$.

\subsection{High Level Overview of the Routing Scheme}\label{sec:high_level_routing}

Our scheme is best described via two special cases; in the first case, $t$ is in an $A$-fragment, and in the second case, the $s$-to-$t$ path in $T_c$ is only via $B$-fragments.
We then describe how to connect between these two cases, essentially by first routing to the $A$-fragment that is nearest to $t$ in $T_c$, and then routing from that $A$-fragment to $t$ (crucially, this route does not contain $A$-fragments).

\paragraph{First Case: $t$ is in an $A$-fragment.}
Suppose an even stronger assumption, that we are actually given a vertex $a^*\in A$ that is in the same fragment as $t$. We will resolve this assumption only at the wrap-up of this section.
The general strategy is to try and follow the $s$-to-$t$ path in the recovery tree $T_c$.
This path is of the form $P_1 \circ e_1 \circ P_2 \circ e_2 \circ \cdots \circ e_\ell \circ P_\ell$, where each $P_i$ is a path in a fragment $X_i$ of $T-c$, and the $e_i$ edges are recovery edges connecting between fragments, so that $X_\ell$ is the fragment of $t$ in $T-c$.
Rather than following this path directly, our goal will be to route from one fragment to the next, through the corresponding recovery edge.

As there are only $O(\bp(G))$ $A$-fragments, every $v\in V$ can store $O(\log n)$ bits for each $A$-fragment.
However, the routing table of $v$ cannot store said information for every color.
To overcome this obstacle, note that in every fragment in $T-c$ (besides the one containing the root $r$), the root of the fragment is connected to its parent via a $c$-colored edge.
We leverage this property, and let the root of every fragment store, for every $a \in A$, the first recovery edge $e(v,a,c)$ on the path from $v$ to $a$ in $T_c$.
Thus, 
when reaching the fragment $X_i$ of $T-c$, we first go up as far as possible, until we hit the root of $X_i$.
In the general case, this is a vertex $v_i$ such that the edge to its parent is of color $c$.
Therefore, $v_i$ stores in its table the next recovery edge $e_i = e(v_i,a^*,c)$ we aim to traverse.
% There is also a special case when $v = r$, which is dealt using the permanent header.
The special case of $v_i = r$ is resolved using the color labels.
The color $c$ stores, for every $a\in A$, the first recovery edge $e(r,a,c)$; at the start of the routing procedure, $s$ extracts the information regarding $e(r,a^*,c)$ and writes it in the header.

So, we discover $e_i$ in $v_i$, and next we use the Thorup-Zwick routing of \Cref{lem:tree-routing} on $T$ to get to the first endpoint of $e_i$.
The path leading us to this endpoint is fault-free (it is contained in the fragment $X_i$).
Then, we traverse $e_i$, and continue in the same manner in the next fragment $X_{i+1}$.

Once we reach the $A$-fragment that contains $a^*$ and $t$, we again use the Thorup-Zwick routing of \Cref{lem:tree-routing} on $T$.
For that we also need $L_T(t)$, which $s$ can learn from the label of $t$ and write in the header at the beginning of the procedure.

\paragraph{Second Case: the $s$-to-$t$ path in $T_c$ is only via $B$-fragments.}
As every vertex $v$ in the $s$-to-$t$ path in $T_c$ is in a $B$-fragment of $T-c$, by \Cref{lem:fragments}, $c\in P(v)$.
Thus, $v$ can store the relevant tree-routing table $R_{T_c}(v)$.
Essentially, every $v$ has to store such routing table for every color in $P(v)$.
Also, since $c \in P(v)$, $t$ can store in its label $L(t)$ the tree-routing label $L_{T_c} (t)$, and the latter can be extracted by $s$ and placed on the header of the message at the beginning of the procedure.
Hence, we can simply route the message using Thorup-Zwick routing scheme of \Cref{lem:tree-routing} on $T_c$.

\paragraph{Putting It Together.}
We now wrap-up the full routing procedure.
If $c\notin P(t)$, then $t$ is connected to $a(t)$, and we get the first case with $a^* = a(t)$, which can be stored in $t$'s label.
Thus, suppose $c\in P(t)$. 
Since $|P(t)|=O(\bp(G))$, the label of $t$ can store $O(\log n)$ bits for every color on $P(t)$, and specifically for the color $c$ of interest.
If $t$ is in an $A$-fragment in $T-c$, then $t$ can pick an arbitrary $A$-vertex in its fragment as $a^*$, and again we reduce to the first case. 
Suppose $t$ is in a $B$-fragment in $T-c$. 
In this case, $t$ sets $a^*$ to be an $A$-vertex from the nearest $A$-fragment to $t$ in $T_c$. 
% Let $t'$ be a nearest vertex to $t$ in $T_c$ that is in an $A$-fragment, and let $a^*$ be an arbitrary $A$-vertex in the fragment of $t'$.
The label of $t$ can store $a^*$ and the first recovery edge from $a^*$ towards $t$ (i.e., $e(a^*,t,c)$)
At the beginning of the procedure, $s$ can find the information regarding $a^*$ and $e(a^*,t,c)$ in $t$'s label, and write it on the message header.
Now, routing from $s$ to the fragment of $a^*$ is by done by the first case, traversing this fragment towards $e(a^*,t,c)$ is done using Thorup-Zwick tree-routing on $T$, and after taking this edge, we can route the message to $t$ according to the second case.

\subsection{Construction of Routing Tables and Labels}
We now formally define the tables and labels of our scheme, by \Cref{alg:tables,alg:vertex-routing-labels,alg:color-routing-labels}.

\begin{algorithm}[h]
\caption{Creating the table $R(v)$ of vertex $v$}\label{alg:tables}
\begin{algorithmic}[1]
% \Require 
% \Ensure 
\State \textbf{store} $R_T (v)$
\State \textbf{store} port number of the edge from $v$ to its parent in $T$
\State $c(v) \gets \text{color of edge from $v$ to its parent in $T$}$ \Comment{undefined if $v = r$}
\State \textbf{store} $c(v)$
\For{each vertex $a \in A$}
    \State \textbf{store} $\FirstRecEdge(v,a,c(v))$
\EndFor
\For{each color $c \in P(v)$}
    \State \textbf{store} $R_{T_c} (v)$
\EndFor
\end{algorithmic}
\end{algorithm}
% Next, the vertex labels are generated by \Cref{alg:vertex-routing-labels}
\begin{algorithm}[h]
\caption{Creating the label $L(v)$ of vertex $v$}\label{alg:vertex-routing-labels}
\begin{algorithmic}[1]
% \Require 
% \Ensure 
\State \textbf{store} $L_T (v)$
\State \textbf{store} $a(v)$, the $A$-endpoint of $P(v)$
\For{each color $c \in P(v)$}
    \State $a(v,c) \gets \text{an $A$-vertex in the nearest $A$-fragment to $v$ in $T_c$}$
    \State \textbf{store} $\FirstRecEdge(a(v,c),v,c)$
    \State \textbf{store} $L_{T_c}(v)$ \label{line:block-end}
\EndFor
\end{algorithmic}
\end{algorithm}
% Finally, the color labels are generated by \Cref{alg:color-routing-labels}:
\begin{algorithm}[h]
\caption{Creating the label $L(c)$ of color $c$}\label{alg:color-routing-labels}
\begin{algorithmic}[1]
% \Require 
% \Ensure 
\For{each vertex $a \in A$}
    \State \textbf{store} $\FirstRecEdge(r,a,c)$
\EndFor\end{algorithmic}
\end{algorithm}

\paragraph{Size Analysis.}
It is easily verified that each \textbf{store} instruction in \Cref{alg:tables,alg:vertex-routing-labels,alg:color-routing-labels} adds $O(\log n)$ bits of storage.
In all of these algorithms, the number of such instructions is $O(|P(v)| + |A|)$, which is $O(\bp(G))$ by \Cref{lem:A-set}.
Hence, the total size of any $R(v)$, $L(v)$ or $L(c)$ is $O(\bp(G) \log n)$ bits.

\subsection{The Routing Procedure}

In the beginning of the procedure, $s$ holds the labels $L(s), L(t)$ and $L(c)$, and should route the message $M$ to $t$ avoiding the color $c$.
As described in \Cref{sec:high_level_routing}, the routing procedure will have two phases.
In the first phase, the message is routed to the fragment of $T-c$ that contains a carefully chosen vertex $a^* \in A$.
In the second phase, it is routed from this fragment to the target $t$.

\paragraph{Initialization at $s$.}
First, $s$ determines the vertex $a^*$ as follows: If $c \in P(t)$, then $a^* = a(t,c)$, which is found in $L(t)$. Otherwise, $a^* = a(t)$, the $A$-endpoint of $P(t)$, which is again found in $L(t)$.
Next, $s$ creates the initial header of the message $M$, that contains:
% The initial header that $s$ creates contains the following information:
\begin{itemize}
    \item The name of the color $c$.
    \item The name of the vertex $a^*$.
    \item The block $\FirstRecEdge(r,a^*,c)$, found in $L(c)$.
    \item The tree-routing label $L_T (t)$, found in $L(t)$.
    \item If $c \in P(t)$, the block $\FirstRecEdge(a^*,t,c)$ and the tree-routing label $L_{T_c}(t)$, found in $L(T)$.
\end{itemize}
This information will permanently stay in the header of $M$ throughout the routing procedure, and we refer to it as the \emph{permanent header}. Verifying that it requires $O(\log n)$ bits is immediate.

\paragraph{First Phase: Routing from $s$ to the Fragment of $a^*$.}
As in \Cref{sec:high_level_routing}, let $e_1, e_2, \dots e_\ell$ be the recovery edges on the $s$-to-$a^*$ path in $T_c$ (according to order of appearance), each connecting between fragments $X_{i-1}$ and $X_i$ of $T-c$. Denote by $v_i$ the root of fragment $X_i$, and thus $e_i = e(v_i, a^*, c)$.
Recall that we aim to route the message through these edges and reach the last fragment $X_\ell$, which is an $A$-fragment containing $a^*$, according to the following strategy:
Upon reaching a fragment $X_i \neq X_\ell$, we go up until we reach its root $v_i$, extract information regarding the next recovery edge $e_i$, and use Thorup-Zwick tree-routing on $T$ in order to get to $e_i$. We then traverse it to reach $X_{i+1}$, and repeat the process.

We now describe this routing procedure formally.
The header of $M$ contains two updating fields:
\begin{itemize}
    \item $M.\UP$: stores a Boolean value
    \item $M.\NEXT$: stores a block referring to the next recovery edge in the path, of the form $\FirstRecEdge(\cdot,a^*,c)$.
    % an edge $e$ with a direction, the port number in this direction, the $L_T (\cdot)$ label of the first vertex of $e$, and a bit stating if the second vertex of $e$ is in $X_\ell$.
    % For convenience, the notation ``$M.\NEXT \gets e$'' means that $M.\NEXT$ is set to contain all the above information.
\end{itemize}
Clearly, this requires $O(\log n)$ bits of storage.
We will maintain the following invariant: 
\begin{itemize}
    \item[(I):] If the message is currently in the fragment $X_i \neq X_\ell$, and $M.\UP = 0$, then $M.\NEXT$ stores information referring to $e_i$.
\end{itemize}

At initialization, $s$ sets $M.\UP \gets True$ and $M.\NEXT \gets \perp$ (a null symbol), which trivially satisfies the invariant (I).
We will use the field $M.\UP$ also to mark that we are still in the first phase of the routing. During the first phase, it will be a valid Boolean value. When we reach the fragment $X_\ell$, we set $M.\UP$ to $\perp$ to notify the beginning of the second phase.
While $M.\UP\neq \perp$ (i.e., during the first phase), upon the arrival of $M$ to a vertex $v$, it executes the code presented in \Cref{alg:first-phase-routing} to determine the next hop.

\begin{algorithm}[h]
\caption{First phase: Routing $M$ from $v$ towards the fragment of $a^*$ (while $M.\UP\neq \perp$)}\label{alg:first-phase-routing}
% upon reaching vertex $v$ (while $M.\UP\neq \perp$)}
\begin{algorithmic}[1]
% \Require 
% \Ensure 
\If{$M.\UP = True$}
    \If{$v \neq r$ and $c(v) \neq c$}
        \State send $M$ through the port to $v$'s parent in $T$, found in $R(v)$
    \Else  
        \State find $\FirstRecEdge(v,a^*,c)$, in permanent header when $v=r$, or in $R(v)$ when $c(v) = c$
        \State $M.\NEXT \gets \FirstRecEdge(v,a^*,c)$
        \State $M.\UP \gets False$ \label{line:up-gets-false}
    \EndIf
\EndIf
\If{$M.\UP = False$}
    % \Comment{$M.\NEXT = e_i$ by (I)}
    \State $x \gets$ first vertex of the edge $e$ found in $M.\NEXT$
    \If{$v \neq x$}
        \State send $M$ in direction of $x$ in $T$, using $L_T (x)$ from $M.\NEXT$, and $R_T (v)$ from $R(v)$
    \Else
        \Comment{$v = x$}
        \If{second vertex of $e$ is in $X_\ell$ (as indicated by $M.\NEXT$)}
            \State $M.\UP \gets \perp$ \Comment{will reach $X_\ell$ in next step, and then done}
        \Else
            \State $M.\UP \gets True$
        \EndIf
        \State send $M$ through the port of the edge $e$ found in $M.\NEXT$
    \EndIf
\EndIf
\end{algorithmic}
\end{algorithm}

Invariant (I) is maintained when setting $M.\UP \gets False$ in \Cref{line:up-gets-false}, since we previously set $M.\NEXT$ to $e(v,a^*,c)$, which is the first recovery edge in the $v$-to-$a^*$ path in $T_c$, and thus, when $v \in X_i$, this edge equals $e_i$.
When the message first reaches a vertex $s' \in X_\ell$, the field $M.\UP$ contains a null value $\perp$, and we start the second phase of the routing procedure, as described next.

\paragraph{Second Phase: Routing from the Fragment of $a^*$ to $t$.}
Here, we use the careful choice of the vertex $a^*$.
The easier case is when $c \notin P(t)$. In this case, $a^* = a(v)$, and $t$ is in the same fragment as $a^*$ by \Cref{lem:fragments}.
Therefore, we can use the Thorup-Zwick routing of \Cref{lem:tree-routing} on $T$ to route the message $M$ from $s'$ to $t$.
Note that the label $L_T (t)$ is found in the permanent header, and that each intermediate vertex $v$ on the path has $R_T (v)$ in its table.

We now treat the case where $c \in P(t)$.
In this case, $a^* = a(t,c)$, which is defined to be an $A$-vertex in the nearest $A$-fragment to $t$ in $T_c$, and hence, $X_\ell$ is that nearest $A$-fragment to $t$ in $T_c$.
% We first note that $X_\ell$ is the nearest $A$-fragment to $t$ in $T_c$; indeed, in this case, $s$ set $a^* = a(t,c)$, which is defined to be an $A$-vertex in the nearest $A$-fragment to $t$ in $T_c$.
Since $c \in P(t)$, the permanent header stores $\FirstRecEdge(a^*,t,c)$, or specifies that the first recovery edge $e(a^*,t,c)$ is undefined, i.e.\ that $a^*$ and $t$ are in the same fragment.
In the latter case, we can act exactly as above and route the message over $T$. So, assume $\FirstRecEdge(a^*,t,c)$ is found.

Let $x$ and $y$ be the first and second vertices of $e(a^*,t,c)$.
Then $\FirstRecEdge(a^*,t,c)$ contains $L_T (x)$, so we can route the message between $s',x \in X_\ell$ using the Thorup-Zwick routing of \Cref{lem:tree-routing} on $T$.
We then traverse the edge $e(a^*,t,c)$ from $x$ to $y$, where the relevant port is stored in the permanent header.
We now make the following important observation:
\begin{lemma}
    If $v$ is a vertex on the $y$-to-$t$ path in $T_c$, then its table $R(v)$ must contain $R_{T_c}(v)$.
\end{lemma}
\begin{proof}
    First, we note that $v$ must be a $B$-fragment.
    Indeed, if $v$ were in an $A$-fragment, then this would be a closer $A$-fragment to $t$ than $X_\ell$ in $T_c$, which is a contradiction.
    Therefore, by \Cref{lem:fragments}, it must be that $c \in P(v)$.
    This means that $R_{T_c}(v)$ is stored in $R(v)$ by \Cref{alg:tables}.
\end{proof}
Note that $L_{T_c}(t)$ is found in the permanent header, as $c \in P(t)$.
Thus, we can route $M$ from $y$ to $t$ along the connecting path in the recovery tree $T_c$, using the Thorup-Zwick routing of \Cref{lem:tree-routing} on $T_c$.
This concludes the routing procedure.
\section{Single Color Fault: Proof of Theorem \ref{THM:SFUB}}\label{sect:1-fault-upper-bound}
\SINGLEFAULTUPPER
\section{Nearest Colored Ancestor Labels}\label{sect:nearest-colored-ancestor-labels}

In this section, we show how the $O(n)$-space, $O(\log \log n)$-query time nearest colored ancestor data structure of \cite{GawrychowskiLMW18} can be used to obtain $O(\sqrt{n} \log n)$-bit labels for this problem.

The labels version of nearest colored ancestor is formally defined as follows.
Given a rooted $n$-vertex forest $T$ with colored vertices, where each vertex $v$ has an arbitrary unique $O(\log n)$-bit identifier $\id(v)$, the goal is to assign short labels to each vertex and color in $T$, so that the $\id$ of the nearest $c$-colored ancestor of vertex $v$ can be reported by inspecting the labels of $c$ and $v$.
The reduction from the centralized setting implies that such a labeling scheme can be used `as is' for connectivity under one color fault.\footnote{
    When constructing the nearest colored ancestor labels in the reduction, we augment the $\id$ of each vertex $v$ with $\cid(v, G-c)$, where $c$ is the color of the tree edge from $v$ to its parent.
}
First, by \Cref{thm:single-fault-lower-bound}, we get:
\begin{corollary}
    Every labeling scheme for nearest colored ancestor in $n$-vertex forests must have label length $\Omega(\sqrt{n})$ bits.
    Furthermore, this holds even for paths, as their ball-packing number is $\Omega(\sqrt{n})$.
\end{corollary}

\begin{remark}
The above lower bound can be strengthened to $\Omega(\sqrt{n} \log n)$.
This is by considering the problem that our scheme \emph{actually} solves: report the minimum $\id$ of a vertex connected to $v$ in $G-c$, from the labels of $v$ and $c$.
For this problem, one can extend the proof of \Cref{thm:single-fault-lower-bound} to show an $\Omega(\sqrt{n} \log n)$-bit lower bound for paths.
The reduction described above shows that a nearest colored ancestor labeling scheme can be used to report such minimum $\id$s.
\end{remark}

The data structure in \cite{GawrychowskiLMW18} can, in fact, be transformed into $O(\sqrt{n} \log n)$-bit labels.
We first briefly explain how this data structure works.
Each vertex $v$ gets two \emph{time-stamps} $\pre(v),\post(v)$, which are the first and last times a DFS traversal in $T$ reaches $v$.
The time-stamps of $c$-colored vertices are inserted to a \emph{predecessor structure} \cite{BoasKZ77, Willard83} for color $c$.
For each time-stamp of a ($c$-colored) vertex $u$, we also store the $\id$ of the nearest $c$-colored (strict) ancestor of $u$.
A query $(v,c)$ is answered by finding the predecessor of $\pre(v)$ in the structure of $c$.
If the result is $\pre(u)$, then $u$ is returned.
If it is $\post(u)$, then the ancestor pointed by $u$ is returned.
Correctness follows by standard properties of DFS time-stamps.
The predecessor structure for $c$ answers queries in $O(\log \log n)$ time, and takes up $O(|V_c|)$ space (in words), where $V_c$ is the set of $c$-colored vertices. The total space is $O(\sum_{c} |V_c|) = O(n)$.

To construct the labels, let $\mathcal{H} = \{c \mid |V_c| \geq \sqrt{n}\}$ be the \emph{highly prevalent} colors, and $\mathcal{R}$ be the rest of the colors.
As there are only $n$ vertices, $|\mathcal{H}| = O(\sqrt{n})$.
We can therefore afford to let each vertex $v$ explicitly store in its label the $\id$ of its nearest $c$-color ancestor, for each $c \in \mathcal{H}$.
To handle the remaining $\mathcal{R}$-colors, we store in the label of each $c \in \mathcal{R}$ the predecessor structure for $c$, which only requires $O(|V_c| \log n) = O(\sqrt{n} \log n)$ bits.
By augmenting the vertex labels with $\pre(\cdot)$ times (requiring only $O(\log n)$ additional bits), we can also answer queries with colors in $\mathcal{R}$.
We obtain:
\begin{corollary}
    There is a labeling scheme for nearest colored ancestor in $n$-vertex forests, with labels of length $O(\sqrt{n} \log n)$ bits.
    Queries are answered in $O(\log \log n)$ time.
\end{corollary}

\section{Limitations of the Lower Bound Technique of Theorem \ref{thm:f-faults-lower-bound}}\label{par:f_faults_LB_barrier}
% In some sense, the proof technique of~\Cref{thm:f-faults-lower-bound} cannot be used to obtain a lower bound stronger than $\tilde{\Omega}(n^{1-1/(f+1)})$.
The `technique' in the proof of \Cref{thm:f-faults-lower-bound} is devising a protocol for $\INDEX(N)$ that uses, \emph{as a black-box}, a labeling scheme for $n$-vertex graphs, by specifying:
\begin{enumerate}
    \item A (`worst-case') graph topology $G$, known in advance to Alice and Bob.
    \item A coloring procedure by which Alice colors her copy of $G$ according to her input.
    \item Which labels are sent from Alice to Bob.
    \item A recovery procedure of Bob, using the received labels for connectivity queries.
\end{enumerate}

Let $S$ be the maximum possible number of labels that Alice sends in an execution of the communication protocol.
We assume that $N=\Theta(n)$, and justify this assumption later.
Then,
\Cref{lem:index_LB} yields a lower bound of $\Omega(n/S)$ on the label length, and
we now argue why this lower bound cannot be stronger than $\tilde{\Omega}(n^{1-1/(f+1)})$.

Let $C$ be the maximum possible number of \emph{color labels} that Alice sends.
Bob can only simulate the failure of color for which he has a label.
So, any color whose label is not sent can be considered non-faulty, and replaced with $\perp$ in the coloring procedure, without affecting the correctness of the protocol.

Consider the simple labeling scheme where each vertex label $L(v)$ stores $\cid(v,G-F)$ for every set $F$ of at most $f$ colors, and a color label simply stores the color's name.
As there are $C$ colors, the maximum length of a label given by this scheme
% when executing the protocol 
is $O(C^f \log n)$ bits.
Applying the lower bound, we get
$\Omega(n/S) = C^f \log n \leq S^f\log n$,
hence $S=\tilde{\Omega}(n^{1/(f+1)})$, and thus $n/S = \tilde{O}(n^{1-1/(f+1)})$, which concludes this argument.

One could hope to encode more than $N=\Theta(n)$ bits using the topology $G$, as the number of bits to encode a specific coloring of $G$ is $\Theta(m \log C)$.
However, the sparsification of \Cref{lem:conn-cert}
show that the colored $G$ has a subgraph $H$ (which depends on the coloring) with $\tilde{O}(n)$ edges, where all connectivity queries give the same answers in $H$ as in $G$.
As Bob's recovery procedure only uses such queries, 
the information he recovers can be encoded by this colored subgraph $H$, and hence contains at most $\tilde{\Theta}(n)$ bits.
% he cannot distinguish between $G$ and $H$, so he can recover at most $\tilde{\Theta}(n)$ bits corresponding to colorings of $H$.

\section{Hitting Set}\label{sect:hitting-set}

The following lemma is well known. The proof, adapted from \cite{AingworthCM96}, is based on the equivalence of the hitting set and set cover problems.
\begin{lemma}[Deterministic Hitting Set]\label{lem:det_hitting_set}  
    Let $V$ be a set of size $n$, and let $S_1, \dots, S_k\subseteq V$, each of size at least $\Delta$.
    Then there is a subset $U \subseteq V$ with $|U| = O((n/\Delta) \log k)$ such that $S_i \cap U \neq \emptyset$ for all $i \in \{1,\dots, k\}$.
    The set $U$ can be found by an efficient deterministic algorithm.
\end{lemma}
\begin{proof}    
For every $v \in V$, let $A_v = \{i \mid v \in S_i\}$,
and consider the set cover instance $\mathcal{A} = \{A_v \mid v \in V\}$.
Take $U$ such that $\{A_u \mid u \in U\}$ is the cover produced by the greedy algorithm for $\mathcal{A}$.
To spell this out:
Initialize $U = \emptyset$.
While there is $i$ with $S_i \cap U = \emptyset$ (i.e., $i \notin \bigcup_{v \in U} A_v$),
choose $u \in V$  that maximizes the quantity
$
| A_u - \bigcup_{v \in U} A_v | = |\{i \mid u \in S_i, S_i \cap U = \emptyset\}|,
$
and update $U \gets U \cup \{u\}$.
This can be implemented in linear time (in the input size) using a bucket queue.

It is known that the greedy solution for set cover approximates the optimal \emph{fractional} solution within a $O(\log k)$ factor \cite{Johnson74,Lovasz75}.
Give weight $w(A_v) = 1/\Delta$ to each set $A_v$.
This is a feasible fractional solution, as for every $i \in \{1,\dots,k\}$,
$
\sum_{v : i \in A_v} w(A_v) = |S_i| / \Delta \geq 1.
$
Its value is  $\sum_{v \in V} w(A_v) = n/\Delta$.
Thus, $|U| = O((n/\Delta) \cdot \log k)$.
\end{proof}

\end{document}